\documentclass[%
a4paper,%
USenglish,%
numberwithinsect]{lipics-v2021}
\pdfoutput=1 %uncomment to ensure pdflatex processing (mandatatory e.g. to submit to arXiv)
\hideLIPIcs  %uncomment to remove references to LIPIcs series (logo, DOI, ...), e.g. when preparing a pre-final version to be uploaded to arXiv or another public repository

\usepackage{bbm}

\theoremstyle{definition}
\newtheorem*{definition*}{Definition}
\theoremstyle{plain}

\newcommand{\Integers}{\mathbb{Z}}

\newcommand{\Ports}{\mathit{P}}
\newcommand{\Neighbors}{\mathit{N}}
\newcommand{\Degree}{\operatorname{deg}}
\newcommand{\CounterPort}[1]{\bar{#1}}
\newcommand{\Siblings}{\operatorname{sibs}}

\providecommand{\Pr}{}\renewcommand{\Pr}{\mathbb{P}}
\providecommand{\Ex}{}\renewcommand{\Ex}{\mathbb{E}}
\newcommand{\GeomDist}[1]{\operatorname{Geom}\!\left(#1\right)}
\newcommand{\UnifDist}{\operatorname{Unif}}

\newcommand{\varDummy}{\texttt{var}}
\newcommand{\varLLE}{\textsf{LLE}}
\newcommand{\varReorientationBit}{\texttt{rort}}
\newcommand{\varLocalOrientation}{\texttt{ort}}
\newcommand{\varStatus}{\texttt{sts}}
\newcommand{\varUnorientatedFlag}{\texttt{uort}}
\newcommand{\varCoin}{\mathsf{coin}}
\newcommand{\varPhaseMode}{\texttt{mod}}
\newcommand{\varEpoch}{\texttt{epc}}
\newcommand{\varTournament}{\texttt{trnt}}
\newcommand{\varProposal}{\texttt{prp}}
\newcommand{\varReceivingFlag}{\texttt{rcv}}
\newcommand{\varAcceptFlag}{\texttt{acpt}}
\newcommand{\varDesignated}{\texttt{dsg}}
\newcommand{\varNextFlip}{\mathsf{nxt}}
\newcommand{\varFlipFlag}{\mathsf{flp}}

\newcommand{\IN}{\mathsf{IN}}
\newcommand{\OUT}{\mathsf{OUT}}
\newcommand{\MATCHED}{\mathsf{MATCHED}}
\newcommand{\UNMATCHED}{\mathsf{UNMATCHED}}
\newcommand{\Proposing}{\mathsf{proposing}}
\newcommand{\Receiving}{\mathsf{receiving}}
\newcommand{\Standby}{\mathsf{standby}}
\newcommand{\TO}{\mathsf{TO}}
\newcommand{\FROM}{\mathsf{FROM}}

\newcommand{\Alg}{\mathtt{Alg}}
\newcommand{\Indicator}{\mathbbm{1}}
\providecommand{\T}{}\renewcommand{\T}{\mathcal{T}}
\providecommand{\True}{}\renewcommand{\True}{\mathsf{true}}
\providecommand{\False}{}\renewcommand{\False}{\mathsf{false}}
\newcommand{\tClean}{t_{\mathsf{clean}}}

\newcommand{\rme}{{\rm e}}
\newcommand{\rmc}{{\rm c}}

\title{Self-Stabilizing Algorithms in the Uniform Port Model} %TODO Please add

\author{Liam Brinker}%
{Technion - Israel Institute of Technology, Israel}%
{liam.brinker@campus.technion.ac.il}%
{}{}
\author{Yuval Emek}%
{Technion - Israel Institute of Technology, Israel}%
{yemek@technion.ac.il}%
{https://orcid.org/0000-0002-3123-3451}%
{Partially supported by an Israel Science Foundation (ISF) grant 730/24.}
\author{Oren Louidor}%
{Technion - Israel Institute of Technology, Israel}%
{olouidor@technion.ac.il}%
{https://orcid.org/0000-0002-8401-6485}%
{Partially supported by an Israel Science Foundation (ISF) grant 3782/25.}
\authorrunning{L. Brinker, Y. Emek, and O. Louidor}
\Copyright{Jane Open Access and Joan R. Public} %TODO mandatory, please use full first names. LIPIcs license is "CC-BY";  http://creativecommons.org/licenses/by/3.0/

\nolinenumbers %uncomment to disable line numbering

\EventEditors{John Q. Open and Joan R. Access}
\EventNoEds{2}
\EventLongTitle{42nd Conference on Very Important Topics (CVIT 2016)}
\EventShortTitle{CVIT 2016}
\EventAcronym{CVIT}
\EventYear{2016}
\EventDate{December 24--27, 2016}
\EventLocation{Little Whinging, United Kingdom}
\EventLogo{}
\SeriesVolume{42}
\ArticleNo{23}
\begin{document}

\setcounter{page}{0}

\maketitle

\begin{abstract}
We introduce a distributed computational model referred to as the
\emph{uniform port} model.
An algorithm operating in this model is defined by means of local automata
associated with the ports (a.k.a.\ half-edges) of the input graph.
The crux of the uniform port model is that a single constant-size finite
automaton is hosted by every port of every graph, making the model
\emph{truly uniform}.
Moreover, since the new model explicitly supports the assignment of (input
and) output labels to the graph's (half-)edges, it facilitates natural
formulations of (half-)edge-labeling problems such as maximal matching and
sinkless orientation, which are outside the expressivity scope of prior
node-centric truly uniform distributed computational models.

The main technical contribution of this paper is the design of efficient
(i.e., with poly-logarithmic runtime) \emph{self-stabilizing} uniform port
algorithms, operating on general graphs, for various fundamental local
symmetry breaking problems, including
maximal independent set,
maximal matching,
sinkless orientation,
and maximal node/edge $k$-coloring.
While efficient self-stabilizing algorithms for local symmetry breaking
problems have been extensively studied in stronger computational models, our
work is the first to demonstrate the existence of such algorithms in a truly
uniform model.
\end{abstract}

\clearpage

\setcounter{page}{0}

\tableofcontents
\clearpage

\setcounter{page}{1}

%%%%%%%%%%%%%%%%%%%%%%%%%%%%%%%%%%%%%%%%%%%%%%%%%%%%%%%%%%%%%%%%%%%%%%%%%%%%%%
\section{Introduction}
\label{sec:intro}
%%%%%%%%%%%%%%%%%%%%%%%%%%%%%%%%%%%%%%%%%%%%%%%%%%%%%%%%%%%%%%%%%%%%%%%%%%%%%%
The notion of uniformity in computational models refers to settings in which
the algorithm designer is oblivious to the size (and often other parameters)
of the problem instances on which the algorithm is executed.
The main advantage of uniform algorithms is that they enable a
one-size-fits-all approach:
the same computational device can be used for \emph{all} problem instances.
In distributed computing, this advantage is amplified as a single execution
often involves a multitude of such devices.
Traditionally, uniform distributed algorithms were studied without accounting
for the memory size of the computational devices, effectively allowing
infinite memory.

The turning point in this regard came when the research community realized
that distributed computing is not limited to computer networks;
rather, interesting distributed processes also arise in networks composed of
devices that are much weaker than silicon-based computers.
A pioneering work in this direction is the paper of Afek et
al.~\cite{AfekABHBBJ2011drosophila} who discovered that a biological process
occurring during the development of the nervous system of the Drosophila
melanogaster is equivalent to solving an instance of the \emph{maximal
independent set (MIS)} problem --- a classic local symmetry breaking problem
that is among the most extensively studied problems in distributed computing.

Motivated by this discovery, as well as by networks of nano-scale mechanical
devices \cite{AkyildizJP2011nanonetworks}, Emek and Wattenhofer
\cite{EmekW2013stone-age} introduced the \emph{stone age} model as an
abstraction for distributed computing in networks composed of weak devices.
In this model, each node hosts the same (randomized) finite state automaton
whose description size is bounded by a universal constant, independently of
any parameter of the input graph, including the number of nodes and their
degrees;
that is, the model is \emph{truly uniform}.

Since its introduction, the stone age model has attracted significant attention
\cite{EmekW2013stone-age,
AfekEK2018leader-selection,
AfekEK2018synergy,
EmekU2020dynamic-sa,
EmekK2021thin,
GiakkoupisZ2023two-state-process,
VacusZ2025minimalist-le},
leading to the development of algorithms for fundamental problems such as MIS
\cite{EmekW2013stone-age,
EmekU2020dynamic-sa,
EmekK2021thin,
GiakkoupisZ2023two-state-process}
and leader election (and variants thereof)
\cite{AfekEK2018leader-selection,
EmekK2021thin,
VacusZ2025minimalist-le}.
A common drawback of most algorithms in this line of work is that they assume
\emph{graceful initialization}, where all automata start the execution in
synchrony from a designated initial state.
This assumption turns out to be too strong for many practical scenarios, where
it is difficult to coordinate between the automata.
Moreover, algorithms that rely on graceful initialization are completely
helpless if the network experiences some transient disturbance that causes the
execution to diverge off course.
In contrast, \emph{self-stabilizing} algorithms
\cite{Dijkstra1974self-stabilization} are required to recover from arbitrary
initial configurations, thus they do not rely on global coordination at
startup and can tolerate transient faults.

Designing efficient self-stabilizing algorithms in truly uniform models turns
out to be challenging and to the best of our knowledge, all existing positive
results are limited in scope.\footnote{%
In the context of local symmetry breaking problems, an algorithm is regarded
as efficient if its runtime is (at most) polylogarithmic in the number of
nodes, see Section~\ref{sec:prelim}.}
These include the MIS algorithm of Emek and Keren \cite{EmekK2021thin} that assumes graphs of bounded diameter and the MIS algorithm of Giakkoupis and Ziccardi \cite{GiakkoupisZ2023two-state-process} whose analysis is applicable
only to Erd\H{o}s-R\'{e}nyi random graphs.\footnote{\label{footnote:conjecture-two-state-mis-process}%
Giakkoupis and Ziccardi \cite{GiakkoupisZ2023two-state-process} conjectured
that a simple self-stabilizing MIS algorithm called the \emph{$2$-state MIS
process} is efficient on general graphs, however this turns out to be wrong,
see Section~\ref{sec:two-state-process}.}

Another shortcoming of the stone age model is limited expressivity:
A constant number of states implies constant-size output labels, thus stone
age algorithms are suitable only for the class of problems whose outputs can
be encoded by assigning constant-size labels to the nodes of the input graph.
For graphs of unbounded degrees, this class is somewhat narrow, excluding some
fundamental distributed computing problems such as \emph{maximal matching
(MM)} and \emph{sinkless orientation (SO)}.
Indeed, a node of high degree cannot encode the selection of a single
incident edge (MM) or the orientation of all incident edges (SO) using only a
constant-size label.

Motivated by the aforementioned shortcomings of the stone age model, in this
paper, we introduce a new truly uniform model, referred to as the
\emph{uniform port (UP)} model.
The key idea behind the UP model is that the (truly uniform) automata are
hosted by the \emph{ports} (i.e., \emph{half-edges}) of the input graph,
rather than by the nodes, so that each port may interact with its counter-port
on the same edge as well as with its sibling ports on the same node (refer to
Section~\ref{sec:prelim} for a formal definition).
Importantly, the interaction with sibling ports is highly restricted:
the automaton's transition function observes only the set of sibling states,
rather than their multiplicities or identities (cf.\ the \emph{set-broadcast}
communication scheme of \cite{HellaJKLLLSV2015weak-models}), ensuring that the
automaton description remains constant-size and independent of node degree.
We emphasize that while prior distributed computational models use the graph
ports for communication, the novelty of the UP model is that the ports are the
computational entities, each running the same constant-size automaton, while
sibling interaction is constrained so as to preserve true uniformity.

Switching from a node-centric point-of-view to a (half-)edge-centric one
significantly increases the expressivity of the model, thus addressing
one major shortcoming of the stone age model.\footnote{%
It is interesting to point out that although the focus of the current paper is
on distributed computing in graphs, the UP model can be naturally extended to
distributed computing in \emph{hypergraphs}.
Here, a port may have multiple counter-ports that share the same edge, as well
as multiple sibling ports that share the same node;
the extended model would apply the set-broadcast communication scheme to both
sets, thus preserving true uniformity.}
Furthermore, the UP model, that enables finer-grained local coordination,
turns out to be better suited for self-stabilization, thus circumventing the
other shortcoming.

%%%%%%%%%%%%%%%%%%%%%%%%%%%%%%%%%%%%%%%
\subparagraph{A Decentralized Variant of a (non-truly uniform) Node-Centric Model.}
%%%%%%%%%%%%%%%%%%%%%%%%%%%%%%%%%%%%%%%
Recall that in the UP model, each node of degree $d$ hosts $d$ constant-memory
automata.
As such, the reader may wonder how the UP model compares with an alternative
distributed computational model where each node of degree $d$ hosts an
automaton with memory size
$\Theta (d)$
(i.e., with
$2^{\Theta (d)}$
states).
Since in both models, each node $v$ has the same total number of memory bits,
it may seem plausible that the two models are ``essentially the same'';
put differently, one may hope that the automata hosted by the ports of $v$
under the UP model can collectively simulate the automaton hosted by node $v$
itself under the alternative model.

This intuition is misleading:
even under ideal conditions (and the conditions of self-stabilizing algorithms
are far from ideal), the limited communication capabilities of the UP model do
not enable the simulation of an automaton with memory size
$\Theta (d)$
by $d$ constant-memory automata.
In fact, the UP model can be viewed as a ``decentralized variant'' of the
aforementioned alternative model, where each automaton with memory size
$\Theta (d)$
is distributed over $d$ constant-memory automata.
This decentralized variant view is the key to the true uniformity of the UP
model and what makes it so appealing:
assembling a ``node device'' from $d$ generic ``port devices'' is much simpler
than having to manufacture (and carry around) different types of ``node
devices'', one for each value of $d$;
the technician installing the network is likely to appreciate this simplicity
(recall that in some of the applications we are interested in, the
``technician'' is a tiny cell organelle).

%%%%%%%%%%%%%%%%%%%%%%%%%%%%%%%%%%%%%%%
\subparagraph{Technical Contribution.}
%%%%%%%%%%%%%%%%%%%%%%%%%%%%%%%%%%%%%%%
Beyond the conceptual contribution of introducing the UP model, we develop the
following (synchronous) self-stabilizing UP algorithms, all applicable to
general graphs, where $n$ denotes the number of nodes:
\begin{itemize}

\item
an MIS algorithm that runs in
$O (\log^{2} n)$
time w.h.p.;\footnote{%
A statement holds \emph{with high probability (w.h.p.)} if it is satisfied with probability at least
$1 - n^{-c}$
for an arbitrarily large constant
$c > 0$.}

\item
an MM algorithm that runs in
$O (\log^{5} n)$
time w.h.p.;

\item
an SO algorithm that runs in
$O (\log^{2} n)$
time w.h.p.;

\item
a maximal node $k$-coloring algorithm that runs in
$O (\log^{2} n)$
time w.h.p.\ for any constant
$k \geq 2$;
and

\item
a maximal edge $k$-coloring algorithm that runs in
$O (\log^{5} n)$
time w.h.p.\ for any constant
$k \geq 2$.

\end{itemize}
We emphasize that to the best of our knowledge, these are the first efficient
self-stabilizing algorithms for any natural local symmetry breaking problem
(on general graphs), operating under a truly uniform model.

%%%%%%%%%%%%%%%%%%%%%%%%%%%%%%%%%%%%%%%
\subparagraph{Paper's Outline.}
%%%%%%%%%%%%%%%%%%%%%%%%%%%%%%%%%%%%%%%
The rest of the paper is organized as follows.
In Section~\ref{sec:prelim}, we present a formal definition of the UP model
together with some preliminary definitions and machinery that serve us in the
subsequent technical sections.
Sections \ref{sec:mis}, \ref{sec:mm}, and \ref{sec:so} are dedicated to our
self-stabilizing MIS, MM, and SO algorithms, respectively;
each of these sections includes a high level presentation of the
corresponding algorithm followed by a detailed ``port-level description'' and
analysis.
The self-stabilizing node/edge maximal $k$-coloring algorithms are developed
in Section~\ref{sec:maximal-coloring}.
In Section~\ref{sec:two-state-process}, we refute the conjecture of Giakkoupis
and Ziccardi \cite{GiakkoupisZ2023two-state-process} regarding the efficiency
(in general graphs) of the $2$-state MIS process (see
footnote~\ref{footnote:conjecture-two-state-mis-process}).
We conclude in Section~\ref{sec:related-work} with a discussion of some
additional related work.

%%%%%%%%%%%%%%%%%%%%%%%%%%%%%%%%%%%%%%%%%%%%%%%%%%%%%%%%%%%%%%%%%%%%%%%%%%%%%%
\section{Preliminaries}
\label{sec:prelim}
%%%%%%%%%%%%%%%%%%%%%%%%%%%%%%%%%%%%%%%%%%%%%%%%%%%%%%%%%%%%%%%%%%%%%%%%%%%%%%

%%%%%%%%%%%%%%%%%%%%%%%%%%%%%%%%%%%%%%%
\subparagraph{Graphs.}
%%%%%%%%%%%%%%%%%%%%%%%%%%%%%%%%%%%%%%%
Consider an undirected graph \(G\).\footnote{%
Unless stated otherwise, all graphs in this paper are assumed to be finite and
simple.
}
Let \(V_{G}\) and \(E_{G}\) denote the node set and edge set of
\(G\), respectively.
Nodes
$u, v \in V$
are regarded as \emph{adjacent} if
$\{ u, v \} \in E$;
edges
$e, f \in E$
are regarded as \emph{adjacent} if
$|e \cap f| = 1$.
For a node
\(v \in V_{G}\),
let
\(\Neighbors_{G}(v)
=
\{ u \in V_{G} : \{ v,u \} \in E_{G}\}\)
denote the set of nodes adjacent to $v$ in $G$ (a.k.a.\ $v$'s
\emph{neighbors}) and let
$\Degree_{G}(v) =  |\Neighbors_{G}(v)|$
denote the \emph{degree} of $v$ in $G$.
The subgraph of $G$ \emph{induced by a node subset}
$U \subseteq V_{G}$
is the graph whose node set is $U$ and whose edge set consists of all edges
$e \in E_{G}$
such that
$e \subseteq U$;
the subgraph of $G$ \emph{induced by an edge subset}
$F \subseteq E_{G}$
is the graph whose edge set is $F$ and whose node set consists of all nodes
$v \in V_{G}$
such that at least one of the edges incident on $v$ is in $F$.

A \emph{half-edge} in $G$ is a pair of the form
$(v, \{ u, v \})$,
where
$v \in V_{G}$
and
$\{ u, v \} \in E_{G}$;
throughout this paper, we use the term \emph{port} as a synonym for a
half-edge.
Given a port
$p = (u, \{ u, v \})$,
we refer to $u$ and
$\{ u, v \}$
as the node and edge, respectively, incident on port $p$.
Let
$\Ports_{G}(v)
=
\{ (v, \{ v, u \}) : u \in \Neighbors_{G}(v) \}$
be the set of ports incident on node
$v \in V_{G}$
and let
$\Ports_{G}
=
\bigcup_{v \in V_{G}} \Ports_{G}(v)$
be the set of all ports in $G$.
Ports incident on the same node are referred to as \emph{siblings};
the set of siblings of a port
$p = (v, \{ u, v\})$
is denoted by
$\Siblings_{G}(p) = \Ports_{G}(v) - \{ p \}$.
The \emph{counter-port} of a port
$p = (v, \{ u, v \})$
is the port
$(u, \{ u, v \})$
incident on the same edge as $p$ and on the node that lies at the other
end of the edge, denoted by $\CounterPort{p}$.

When the underlying graph $G$ is clear from the context, we may omit the
subscript \(G\) and write
\(V\),
\(E\),
\(\Neighbors(v)\),
\(d(v)\),
\(\Ports(v)\),
\(\Ports\),
and
$\Siblings(p)$
instead of
\(V_{G}\),
\(E_{G}\),
\(\Neighbors_{G}(v)\),
\(d_{G}(v)\),
$\Ports_{G}(v)$,
\(\Ports_{G}\),
and
$\Siblings_{G}(p)$,
respectively.
Unless stated otherwise, the number of nodes and edges in \(G\) are denoted by
\(n = |V|\)
and
\(m = |E|\),
respectively.

%%%%%%%%%%%%%%%%%%%%%%%%%%%%%%%%%%%%%%%
\subparagraph{The Uniform Port Model.}
%%%%%%%%%%%%%%%%%%%%%%%%%%%%%%%%%%%%%%%
Consider a graph problem $\mathcal{P}$ defined over a set $\mathcal{O}$ of
output labels.
When invoked on an undirected graph $G$, a distributed algorithm for
$\mathcal{P}$ that operates under the \emph{uniform port (UP)} model
associates a (randomized) \emph{local automaton} with each port
\(p \in \Ports\).\footnote{%
In the context of the UP model, it is assumed that the input graph does not
include nodes of degree $0$.}
The syntax of this local automaton is given by the \(3\)-tuple
\[
\Alg
\, = \,
\langle Q, \omega, \delta \rangle
\, ,
\]
where
\(Q\) is a finite set of states;
\( \omega: Q \rightarrow \mathcal{O} \)
is a function that maps each state \( q \in Q \) to an output label
\(\omega(q) \in \mathcal{O}\);
and
$\delta: Q \times Q \times 2^{Q} \rightarrow \Delta(Q)$
is a (randomized) state transition function to be explained
soon.\footnote{%
Here, $\Delta(Q)$ denotes the collection of probability distributions over
$Q$.}
It is important to note that the description of
\(\Alg = \langle Q, \omega, \delta \rangle\)
is assumed to be fixed, independently of any parameter of the input graph
\(G\).
At the risk of slightly abusing the notation, we use $\Alg$ for both the
syntax of the local automata and the distributed algorithm defined by the
collection of the local automata associated with the ports in $\Ports$ (to be
explained soon);
our intention will be clear from the context.

A \emph{configuration} of algorithm $\Alg$ is a function
$C : \Ports \rightarrow Q$
that assigns a state
$C(p) \in Q$
to (the local automaton associated with) each port
$p \in \Ports$.
The execution of $\Alg$ proceeds in synchronous rounds, where round
$t \in \Integers_{\geq 0}$
spans the time interval
$[t, t + 1)$,
denoting the configuration of $\Alg$ at time $t$ by $C^{t}$.
Starting from an initial configuration $C^{0}$, the execution advances
according to the following inductive mechanism:
Assuming that the configuration $C^{t}$ has already been constructed, the
configuration
$C^{t + 1}$
is constructed so that for each port
$p \in \Ports$,
the state
$C^{t + 1}(p)$
is obtained by picking
\[
C^{t + 1}(p)
\, \sim \,
\delta \left(
C^{t}(p),
C^{t}(\CounterPort{p}),
\left\{ C^{t}(p') : p' \in \Siblings(p) \right\}
\right)
\, .
\]
That is, to obtain the next state
$C^{t + 1}(p)$
of (the local automaton associated with) port $p$, we apply the state
transition function $\delta$ to the
current state
$C^{t}(p)$
of $p$,
the current state
$C^{t}(\CounterPort{p})$
of $p$'s counter-port,
and
the set
$\{ C^{t}(p') : p' \in \Siblings(p) \}$
of current states of $p$'s siblings;
the next state
$C^{t + 1}(p)$
of $p$ is then picked from the probability distribution determined by
$\delta$.
We emphasize that the third argument of $\delta$ is a set, rather than a
multiset or a vector, of states;
it is this model choice that decouples the automaton from the node degrees,
thus allowing the description size of $\delta$ (and $\Alg$) to be a universal
constant, making the UP model truly uniform.

An \emph{output assignment} of $G$ is a function
$Y : \Ports \rightarrow \mathcal{O}$
that assigns an output label $Y(p)$ to each port
$p \in \Ports$.
The output assignment associated with a configuration
$C : \Ports \rightarrow Q$
of $\Alg$, denoted by
$\omega(C)$, is the function
$Y : \Ports \rightarrow \mathcal{O}$
defined so that
$Y(p) = \omega(C(p))$
for each
$p \in \Ports$.\footnote{%
It is often convenient to augment the output label set $\mathcal{O}$ with the
designated ``nil symbol'' $\bot$ so that
$\omega(C(p)) = \bot$
indicates that port
$p \in \Ports$
is \emph{undecided} in the output assignment associated with configuration
$C$.%
}
We say that the execution
$\eta = \{ C^{t}\}_{t \geq 0}$
of $\Alg$ \emph{stabilizes} to an output assignment
$Y : \Ports \rightarrow \mathcal{O}$
by time
$t \in \Integers_{\geq 0}$
if
$\omega(C^{t'}) = Y$
for every
$t' \geq t$.

Let
$\mathcal{Y}_{\mathcal{P}} \subseteq \mathcal{O}^{\Ports}$
be the set of legal output assignments for problem $\mathcal{P}$ on $G$.
Algorithm $\Alg$ is said to be \emph{self-stabilizing} if for every undirected
graph $G$ and for every initial configuration $C^{0}$, it is guaranteed that
the execution of $\Alg$ stabilizes to some output assignment
$Y \in \mathcal{Y}_{\mathcal{P}}$
by time $t$ with probability that goes to $1$ as
$t \rightarrow \infty$.
The \emph{runtime} (a.k.a.\ \emph{stabilization time}) of a self-stabilizing
algorithm $\Alg$ on $G$, starting from $C^{0}$, is defined to be the minimum
$t$ such that the execution of $\Alg$ stabilizes to a legal output assignment
by time $t$.
Notice that the runtime is a random variable, depending on the coin tosses of
the local automata, and we typically aim to bound it w.h.p.\ as a function of
$n$.
In particular, $\Alg$ is said to be \emph{efficient} if its runtime is bounded
by
$\log^{O (1)} n$
w.h.p.

For the sake of clarity, the UP algorithms developed in this paper are
presented in a more human-readable ``procedural description'', assuming that
the ports maintain local variables that are updated from one round to the next
according to a specified logic;
translating these descriptions to the syntax of the UP model, as formulated
above, is a straightforward (though possibly tedious) task.
We stick to the convention that $\varDummy_{p}$ denotes the local variable
$\varDummy$ maintained by a port
$p \in \Ports$
and when presenting the actions of a UP algorithm during ``the current round'',
we denote the value of $\varDummy_{p}$ at the end of the round (i.e., the
beginning of the next round) by $\varDummy^{+}_{p}$.
When analyzing a UP algorithm, we denote the value of a variable $\varDummy$
at time
$t \in \Integers_{\geq 0}$
by $\varDummy^{t}$;
we note that the individual ports do not know $t$, however, this notation is
well defined in the scope of the analysis.

%%%%%%%%%%%%%%%%%%%%%%%%%%%%%%%%%%%%%%%
\subparagraph{Local Leader Election.}
%%%%%%%%%%%%%%%%%%%%%%%%%%%%%%%%%%%%%%%
While the ports of a node
$v \in V$
are anonymous (and unordered), it is often convenient to mark one port
$p \in \Ports(v)$
as designated from the rest of the ports in $\Ports(v)$.
A basic UP mechanism that does that (and is used by all our algorithms this
way or another) is called \emph{local leader election (LLE)}.
For the purpose of the current description of the LLE mechanism, assume that
each port
$p \in \Ports(v)$
maintains a variable
$\varLLE_{p} \in \{ 0, 1 \}$;
the UP algorithms presented in the sequel use different names for this
variable, but the principle is the same.
Port $p$ updates variable $\varLLE_{p}$ in each round according to the
following simple rule:
\begin{itemize}

\item
If
$\varLLE_{p} = 1$
and there exists
$p' \in \Siblings(p)$
such that
$\varLLE_{p'} = 1$,
then $p$ picks
$\varLLE_{p} \sim \UnifDist(\{ 0, 1 \})$.

\item
Otherwise, if
$\varLLE_{p'} = 0$
for all
$p' \in \Siblings(p) \cup \{ p \}$,
then $p$ assigns
$\varLLE_{p} \gets 1$.

\item
Otherwise, $p$ keeps the current value of $\varLLE_{p}$ (formally,
$\varLLE_{p} \gets \varLLE_{p}$).

\end{itemize}
Put differently, as long as there are multiple ports
$p \in \Ports(v)$
with
$\varLLE_{p} = 1$,
each such port tosses a fair coin to decide whether to reset
$\varLLE_{p} \gets 0$;
and if all ports
$p \in \Ports(v)$
have
$\varLLE_{p} = 0$,
then they all set
$\varLLE_{p} \gets 1$.

Clearly, the mechanism stabilizes when, and only when, there exists exactly
one port
$p \in \Ports(v)$
such that
$\varLLE_{p} = 1$;
we refer to such a port $p$ as the \emph{leader port} of node $v$.
The following lemma, established in Section~\ref{sec:mm:analysis} as part
of the analysis of our self-stabilizing MM algorithm (see
Lemma~\ref{lem:mm:analysis:geometric-tournament-length-upper-bound}),
guarantees that this happens sufficiently fast.

\begin{lemma}
\label{lem:prelim:up-model:local-port-election}
Assume that the LLE mechanism is invoked in a node
$v \in V$
at time
$t_{0} \in \Integers_{\geq 0}$.
There exists a time
$t \geq t_{0}$
such that $v$ admits a leader port
$p \in \Ports(v)$
at time $t$ and $p$ remains the leader port of $v$ subsequently (as long
as the LLE mechanism is not re-invoked in $v$) a.s.\footnote{%
A statement holds \emph{almost surely (a.s.)} if it is satisfied with probability $1$.}
Moreover,
$t \leq t_{0} + O (\log^{2} n)$
w.h.p.
\end{lemma}

%%%%%%%%%%%%%%%%%%%%%%%%%%%%%%%%%%%%%%%
\subparagraph{Miscellaneous.}
%%%%%%%%%%%%%%%%%%%%%%%%%%%%%%%%%%%%%%%
Throughout,
$c > 0$
denotes a constant whose value may depend on the context.
Unless stated otherwise, we stick to the convention that
$\log x = \log_{2} x$.
The (discrete or continuous) uniform distribution over a set $S$ is denoted by
$\UnifDist(S)$
and the geometric distribution, counting the number of independent Bernoulli
trials until (including) the first success, where each trial succeeds with
probability
$0 < p \leq 1$,
is denoted by
$\GeomDist{p}$.
We make an extensive use of the following well known lemma (a proof is
provided in Appendix~\ref{apx:probabilistic-progress} for completeness).

\begin{lemma}
\label{lem:prelim:probabilistic-progress}
Let
$X_{0}, X_{1}, \dots$
be random variables over
$\Integers_{\geq 0}$
that satisfy the following two conditions for every
$i > 0$
a.s.:
(1)
$X_{i} \leq X_{i - 1}$;
and
(2)
$\Ex (X_{i} \mid X_{i - 1}) \leq r X_{i - 1}$
for some fixed parameter
$0 < r < 1$.
Let
$x_{0} = \Ex (X_{0})$
and define the random variable
$T = \min \{ i \geq 0 : X_{i} = 0 \}$.
Then,
$\Pr (T > \lceil \log_{1 / r} x_{0} \rceil + j) \leq r^{j}$
for every
$j \geq 0$.
\end{lemma}

%%%%%%%%%%%%%%%%%%%%%%%%%%%%%%%%%%%%%%%%%%%%%%%%%%%%%%%%%%%%%%%%%%%%%%%%%%%%%%
\section{Maximal Independent Set}
\label{sec:mis}
%%%%%%%%%%%%%%%%%%%%%%%%%%%%%%%%%%%%%%%%%%%%%%%%%%%%%%%%%%%%%%%%%%%%%%%%%%%%%%
A node subset
$S \subseteq V$
is said to be \emph{independent} if
$\{ u, v \} \notin E$
for every
$u, v \in S$.
A \emph{maximal independent set (MIS)} is an independent set
$S \subseteq V$
that is maximal in the sense that $S'$ is not independent for every
$S \subset S' \subseteq V$.
In the MIS problem, the goal is to partition the node set $V$ into
\emph{$\IN$-nodes} and \emph{$\OUT$-nodes} such that the set of $\IN$-nodes
forms an MIS.
Under the UP model, this is translated to using
$\{ \IN, \OUT, \bot \}$
as the output label set so that a node
$v \in V$
is considered to be an $\IN$-node (resp., an $\OUT$-node) in a
configuration $C$ if
$\omega(C(p)) = \IN$
(resp.,
$\omega(C(p)) = \OUT$)
for all
$p \in \Ports(v)$;
node $v$ is regarded as \emph{undecided} in $C$ if it is neither an $\IN$-node
nor an $\OUT$-node.

\begin{theorem}
\label{thm:mis}
There exists a self-stabilizing UP algorithm that solves the MIS problem and
stabilizes in $O(\log^{2} n)$ time w.h.p.
\end{theorem}

The algorithm promised in Theorem~\ref{thm:mis} is presented in
Section~\ref{sec:mis:algorithm}, first at a high level and then, from the
perspective of the individual ports.
Section~\ref{sec:mis:analysis} is then dedicated to proving the correctness of
our algorithm and bounding its stabilization time, thus establishing
Theorem~\ref{thm:mis}.

%%%%%%%%%%%%%%%%%%%%%%%%%%%%%%%%%%%%%%%
\subsection{The MIS Algorithm}
\label{sec:mis:algorithm}
%%%%%%%%%%%%%%%%%%%%%%%%%%%%%%%%%%%%%%%

%%%%%%%%%%%%%%%%%%%%%%%%%%%%%%%%%%%%%%%
\subsubsection*{High-Level Description}
%%%%%%%%%%%%%%%%%%%%%%%%%%%%%%%%%%%%%%%
Our MIS algorithm maintains an orientation of the graph edges and uses this
orientation to determine which nodes join the MIS (i.e., become $\IN$-nodes).
In each round, every node generates a \emph{reorientation bit} from the set
$\{ 0, 1 \}$.
Following that, each edge
$e = \{ u, v \} \in E$
is oriented according to the reorientation bits of its endpoints based on the
following rule:
if the reorientation bits of $u$ and $v$ differ, then $e$ is oriented toward the
endpoint whose reorientation bit is $1$;
otherwise, the edge keeps its previous orientation (if one exists).
Undecided nodes sample their reorientation bit uniformly at
random from
$\{ 0, 1 \}$,
whereas $\IN$-nodes and $\OUT$-nodes fix their reorientation bit to $1$ and
$0$, respectively.

Consider an undecided node
$v \in V$.
Node $v$ becomes an $\IN$-node in round $t$ if all its incident edges are
oriented inward (i.e., $v$ is a sink) at time $t$ and the sampled reorientation
bit of $v$ in round $t$ is $1$;
the latter condition guarantees that following the reorientation step of round
$t$, all edges incident on $v$ remain oriented inward.
Once $v$ becomes an $\IN$-node, its reorientation bit is fixed to $1$, thus
ensuring that $v$ remains an $\IN$-node.
Node $v$ becomes an $\OUT$-node if it has an adjacent $\IN$-node.
Once $v$ becomes an $\OUT$-node, its reorientation bit is fixed to $0$, thus
$v$ yields the orientation of its incident edges to its neighbors.

\begin{remark*}
The reader may wonder about the similarity between our MIS algorithm and
Luby's MIS algorithm \cite{Luby1986mis}, in particular the adaptations thereof
developed by M\'{e}tivier, Robson, Saheb-Djahromi, and Zemmari \cite{MetivierRSZ2011bit-complexity} and by
Emek and Wattenhofer \cite{EmekW2013stone-age}.
We argue that the similarity is minimal:
Luby's MIS algorithm works in phases, so that in each phase, the undecided
nodes ``compete'' among themselves over the right to become $\IN$-nodes.
To this end, each undecided node $v$ picks a random number $r_{v}$ from a
sufficiently large set and becomes an $\IN$-node if
$r_{v} > r_{u}$
for all undecided neighbors $u$ of $v$.
In Luby's original algorithm \cite{Luby1986mis}, node $v$ sends $r_{v}$ in
a single message, thus the competition's outcome is determined within a single
round.
As the algorithms of
\cite{MetivierRSZ2011bit-complexity} 
and \cite{EmekW2013stone-age} use constant size messages, node $v$
sends the value of $r_{v}$ over multiple rounds, essentially communicating
$r_{v}$ one bit at a time.
The key point in this regard is that if $i$ is the index of the most
significant bit in which $r_{v}$ and $r_{u}$ differ so that
$r_{v}(i') = r_{u}(i')$
for all
$1 \leq i' < i$
and
$r_{v}(i) > r_{u}(i)$,
then the outcome of the competition between $u$ and $v$ is determined once $v$
and $u$ exchange the values of $r_{v}(i)$ and $r_{u}(i)$;
there is no point for node $u$ to keep sending the values of $r_{u}(i')$ for
$i' > i$
to $v$ as those values do not affect the outcome of the current phase's competition.

In contrast, our MIS algorithm does not work in phases (as discussed in
Section~\ref{sec:mm:algorithm}, phases raise significant synchronization
challenges for self-stabilizing algorithms).
More importantly, the outcome of the competition between nodes $v$
and $u$ is never fully determined as long as both nodes are undecided:
regardless of the history, a single round in which $v$ and $u$ pick $0$ and
$1$, respectively, as their reorientation bits suffices to reorient edge
$\{ u, v \}$
toward $u$, thus marking $u$ as the current ``front runner'' of the
competition between $u$ and $v$.
This ``erratic behavior'' turns out to be crucial for the fast stabilization
of our algorithm.
\end{remark*}

%%%%%%%%%%%%%%%%%%%%%%%%%%%%%%%%%%%%%%%
\subsubsection*{Port-Level Description}
%%%%%%%%%%%%%%%%%%%%%%%%%%%%%%%%%%%%%%%
For the purpose of orienting the graph edges, each port
$p \in \Ports$
maintains a variable
$\varLocalOrientation_{p} \in \{ 0, 1 \}$,
referred to as the \emph{local orientation} of $p$.
The local orientation variables determine the edge orientations as follows:
An edge
$e = \{ v_{1}, v_{2} \} \in E$
is considered to be oriented toward $v_{i}$,
$i \in \{ 1, 2 \}$,
if
$\varLocalOrientation_{(v_{i}, \{ v_{1}, v_{2} \})} = 1$
and
$\varLocalOrientation_{(v_{3 - i}, \{ v_{1}, v_{2} \})} = 0$;
edge $e$ is considered to be unoriented if
$\varLocalOrientation_{(v_{1}, \{ v_{1}, v_{2} \})} =
\varLocalOrientation_{(v_{2}, \{ v_{1}, v_{2} \})}$,
that is, the edge is neither oriented toward $v_{1}$ nor oriented toward
$v_{2}$.

Related to the local orientation variables, each
port
$p \in \Ports$
maintains a Boolean variable
$\varUnorientatedFlag_{p} \in \{ \True, \False \}$,
referred to as the \emph{unoriented flag} of $p$.
Port $p$ updates this variable in each round by setting
$\varUnorientatedFlag_{p} \gets \True$
if and only if
$\varLocalOrientation_{p} = \varLocalOrientation_{\CounterPort{p}}$,
that is, if and only if the edge shared by $p$ and its counter-port
$\CounterPort{p}$ is unoriented.
The unoriented flag variables have an important role in the algorithm:
they allow port $p$ to observe whether all edges incident on its siblings are
oriented (notice that the local orientation variables alone do not fulfill
this task).

The mechanism that dictates the edge orientation is controlled by a variable
$\varReorientationBit_{p} \in \{ 0, 1 \}$,
referred to as the \emph{reorientation bit} of port $p$, that each port
$p \in \Ports$
maintains.
Specifically, variable $\varLocalOrientation_{p}$ is updated in each round
according to the following simple rule:
if
$\varReorientationBit_{p} \neq \varReorientationBit_{\CounterPort{p}}$,
then
$\varLocalOrientation_{p} \gets \varReorientationBit_{p}$;
otherwise
($\varReorientationBit_{p} = \varReorientationBit_{\CounterPort{p}}$),
the value of $\varLocalOrientation_{p}$ remains unchanged (formally,
$\varLocalOrientation_{p} \gets \varLocalOrientation_{p}$).
The update rule of the reorientation bit variables themselves is presented
shortly;
for now, it suffices to assume that port $p$ updates
$\varReorientationBit_{p}$ in each round by adopting a random coin toss shared
by all ports incident on the same node (the implementation of this shared coin
toss is also presented shortly).

The output label associated with port $p$ is determined based on the value of
the designated variable
$\varStatus_{p} \in \{ \IN, \OUT, \OUT^{*}, \bot \}$,
referred to as the \emph{status} of $p$.
Specifically, the output label associated with $p$ is
$\IN$, $\OUT$, or $\bot$
if
$\varStatus_{p} = \IN$,
$\varStatus_{p} \in \{ \OUT, \OUT^{*} \}$,
or
$\varStatus_{p} = \bot$,
respectively.
Variable $\varStatus_{p}$ is updated in each round according to the following
rule:
\begin{itemize}

\item
Assign
$\varStatus_{p} \gets \IN$
if
(I)
$\varLocalOrientation_{p} = 1$;
(II)
$\varLocalOrientation_{\CounterPort{p}} = 0$;
(III)
$\varUnorientatedFlag_{p'} = \False$
for all
$p' \in \Siblings(p) \cup \{ p \}$;
and
(IV)
$\varReorientationBit_{p'} = 1$
for all
$p' \in \Siblings(p) \cup \{ p \}$.

\item
Otherwise, assign
$\varStatus_{p} \gets \OUT^{*}$
if
$\varStatus_{\CounterPort{p}} = \IN$.

\item
Otherwise, assign
$\varStatus_{p} \gets \OUT$
if
there exists
$p' \in \Siblings(p)$
such that
$\varStatus_{p'} = \OUT^{*}$.

\item
Otherwise, assign
$\varStatus_{p} \gets \bot$.

\end{itemize}

Having established the update rule of the status variables, we can now get
back to the mechanism that dictates the values of the reorientation bit
variables.
This mechanism depends on a leader port
$p^{*} = p^{*}(v) \in \Ports(v)$
elected in each node
$v \in V$
by applying the LLE mechanism (see Section~\ref{sec:prelim}).
Port $p^{*}$ maintains a variable
$\varCoin_{p^{*}} \in \{ 0, 1 \}$,
referred to as the \emph{coin toss} of $v$, and updates this variable in each
round by picking
$\varCoin_{p^{*}} \sim \UnifDist(\{ 0, 1 \})$.
Based on that, ports
$p \in \Ports(v)$
update the variables
$\varReorientationBit_{p}$
in each round according to the following rule:\footnote{%
As our algorithm is self-stabilizing, the behavior of the ports in $\Ports(v)$ prior to the election of the leader port $p^{*}$ is irrelevant and does not affect the analysis.}
\begin{itemize}

\item
Assign
$\varReorientationBit_{p} \gets 1$
if
(I)
$\varUnorientatedFlag_{p'} = \False$
for all
$p' \in \Siblings(p) \cup \{ p \}$;
and
(II)
$\varStatus_{p}^{+} = \IN$
(recall that $\varStatus_{p}^{+}$ denotes the value of $\varStatus_{p}$ at the
end of the round).

\item
Otherwise, assign
$\varReorientationBit_{p} \gets 0$
if
(I)
$\varUnorientatedFlag_{p'} = \False$
for all
$p' \in \Siblings(p) \cup \{ p \}$;
and
(II)
there exists a port
$p' \in \Siblings(p) \cup \{ p \}$
such that
$\varStatus_{p'} = \OUT^{*}$.

\item
Otherwise, assign
$\varReorientationBit_{p} \gets \varCoin_{p^{*}}$.

\end{itemize}

%%%%%%%%%%%%%%%%%%%%%%%%%%%%%%%%%%%%%%%
\subsection{Analysis}
\label{sec:mis:analysis}
%%%%%%%%%%%%%%%%%%%%%%%%%%%%%%%%%%%%%%%
Throughout the analysis, we fix an arbitrary initial configuration $C^{0}$ and
consider the (random) execution
$\eta = \{ C^{t} \}_{t \geq 0}$
of the MIS algorithm starting from $C^{0}$.
A time
$t \in \Integers_{> 0}$
is said to be \emph{clean} with respect to $\eta$ if the following conditions
hold:
\begin{enumerate}

\item
\label{mis:analysis:clean-condition:leaders}
\label{mis:analysis:clean-condition:first}
Every node
$v \in V$
admits a unique leader port
$p^{*} = p^{*}(v) \in \Ports(v)$
in
$C^{t}$.

\item
\label{mis:analysis:clean-condition:oriented-edges}
All edges are oriented in $C^{t}$,
i.e.,
$\varLocalOrientation_{(v_{1}, e)}^{t}
\neq
\varLocalOrientation_{(v_{2}, e)}^{t}$
for every
$e = \{ v_{1}, v_{2} \} \in E$.

\item
\label{mis:analysis:clean-condition:unoriented-flags}
All unoriented flags are down in $C^{t}$, i.e.,
$\varUnorientatedFlag_{p}^{t} = \False$
for every
$p \in \Ports$.

\item
\label{mis:analysis:clean-condition:in-ports}
If port
$p \in \Ports$
satisfies
$\varStatus_{p}^{t} = \IN$,
then
(I)
$\varLocalOrientation_{p}^{t} = 1$;
(II)
$\varReorientationBit_{p}^{t} = 1$;
and
(III)
$\varStatus_{p'}^{t} = \IN$
for all
$p' \in \Siblings(p)$.

\item
\label{mis:analysis:clean-condition:out-star-ports}
If port
$p \in \Ports$
satisfies
$\varStatus_{p}^{t} = \OUT^{*}$,
then
$\varStatus_{\CounterPort{p}}^{t} = \IN$.

\item
\label{mis:analysis:clean-condition:out-ports}
\label{mis:analysis:clean-condition:last}
If port
$p \in \Ports$
satisfies
$\varStatus_{p}^{t} = \OUT$,
then there exists
$p' \in \Siblings(p)$
such that
$\varStatus_{p'}^{t} = \OUT^{*}$.

\end{enumerate}
The following lemma allows us to designate a clean suffix of execution $\eta$;
the subsequent analysis is then dedicated to analyzing that suffix.

\begin{lemma}
\label{lem:mis:analysis:clean-time}
There exists a  time
$\tClean \in \Integers_{> 0}$
such that all times
$t \geq \tClean$
are clean with respect to $\eta$ a.s.
Moreover,
$\tClean \leq O (\log^{2} n)$
w.h.p.
\end{lemma}
\begin{proof}
Lemma~\ref{lem:prelim:up-model:local-port-election} guarantees that there
exists a time
$t^{1} \in \Integers_{> 0}$
such that every node
$v \in V$
admits a (fixed) leader port
$p^{*} = p^{*}(v) \in \Ports(v)$
at all times
$t \geq t^{1}$
and
$t^{1} \leq O (\log^{2} n)$
w.h.p.;
condition hereafter on this event.

Consider an edge
$e = \{ v_{1}, v_{2} \} \in E$
and let
$p_{1} = (v_{1}, e)$
and
$p_{2} = (v_{2}, e)$.
The update rule of the local orientation variables ensures that once $e$
becomes oriented, it remains oriented at all subsequent times.
By the update rule of the unoriented flag variables, for every time
$t \geq t^{1}$,
if $e$ is unoriented at time $t$, then
$\varUnorientatedFlag_{p_{i}}^{t} = \True$,
$i \in \{ 1, 2 \}$,
and thus, the update rule of the reorientation bit variables ensures that
$\varReorientationBit_{p_{i}}^{t + 1} = \varCoin_{p^{*}(v_{i})}^{t}$.
Therefore, edge $e$ becomes oriented during round
$t + 1$
with probability
$1 / 2$,
independently.
We conclude that there exists a time
$t^{2} \in \Integers_{> 0}$
such that all edges are oriented at all times
$t \geq t^{2}$
and
$t^{2} \leq t^{1} + O (\log n)$
w.h.p.;
condition hereafter on this event.

By definition, conditions \ref{mis:analysis:clean-condition:leaders} and
\ref{mis:analysis:clean-condition:oriented-edges} are satisfied at all times
$t \geq t^{2}$.
Having established that
$t^{2} \leq t^{1} + O (\log n) \leq O (\log^{2} n)$,
we complete the proof by showing that conditions
\ref{mis:analysis:clean-condition:unoriented-flags},
\ref{mis:analysis:clean-condition:in-ports},
\ref{mis:analysis:clean-condition:out-star-ports}, and
\ref{mis:analysis:clean-condition:out-ports} hold (deterministically) at all
times
$t > t^{2} + 3$.

Clearly, condition~\ref{mis:analysis:clean-condition:unoriented-flags} holds
from time
$t^{2} + 1$
onward as all edges are oriented from time
$t^{2}$
onward.
Fix some time
$t > t^{2} + 1$
and consider a node
$v \in V$
and a port
$p \in \Ports(v)$
such that
$\varStatus_{p}^{t} = \IN$.
As
$t > t^{2} + 1$,
we know that all edges incident on $v$ are oriented at time
$t - 1$
and that
$\varUnorientatedFlag_{p'}^{t - 1} = \False$
for all
$p' \in \Siblings(p) \cup \{ p \}$.
Since port $p$ adopts the $\IN$ status during round
$t - 1$,
it follows, by the update rule of the status variables, that all edges
incident on $v$ are oriented toward $v$ at time
$t - 1$
and that
$\varReorientationBit_{p'}^{t - 1} = 1$
for all
$p' \in \Siblings(p) \cup \{ p \}$.
Therefore, all ports
$p' \in \Ports(v)$
adopt the $\IN$ status and set
$\varReorientationBit_{p'} \gets 1$
during round
$t - 1$
and the update rule of the local orientation variables ensures that all edges
incident on $v$ remain oriented toward $v$ at time $t$.
Hence, condition~\ref{mis:analysis:clean-condition:in-ports} holds at
time $t$.
By the update rules of the local orientation, reorientation bit, and status
variables, we conclude that
\begin{equation}
\label{eq:mis:analysis:in-remains-in}
\varStatus_{p}^{t} = \IN
\quad\Longrightarrow\quad
\varStatus_{p}^{t + 1} = \IN
\end{equation}
for every port
$p \in \Ports$
and time
$t > t^{2} + 1$.

Fix some time
$t > t^{2} + 2$
and consider a port
$p \in \Ports$
such that
$\varStatus_{p}^{t} = \OUT^{*}$.
Since port $p$ adopts the $\OUT^{*}$ status during round
$t - 1$,
it follows, by the update rule of the status variables, that
$\varStatus_{\CounterPort{p}}^{t - 1} = \IN$.
Condition~\ref{mis:analysis:clean-condition:out-star-ports} holds at time $t$
due to \eqref{eq:mis:analysis:in-remains-in} implying that
$\varStatus_{\CounterPort{p}}^{t} = \IN$.
Furthermore, the same argument guarantees that
\begin{equation}
\label{eq:mis:analysis:out-star-remains-out-star}
\varStatus_{p}^{t} = \OUT^{*}
\quad\Longrightarrow\quad
\varStatus_{p}^{t + 1} = \OUT^{*}
\end{equation}
for every port
$p \in \Ports$
and time
$t > t^{2} + 2$.

Finally, fix some time
$t > t^{2} + 3$
and consider a port
$p \in \Ports$
such that
$\varStatus_{p}^{t} = \OUT$.
Since port $p$ adopts the $\OUT$ status during round
$t - 1$,
it follows, by the update rule of the status variables, that there exists a
port
$p' \in \Siblings(p)$
such that
$\varStatus_{p'}^{t - 1} = \OUT^{*}$.
Condition~\ref{mis:analysis:clean-condition:out-ports} holds at time $t$
due to \eqref{eq:mis:analysis:out-star-remains-out-star} implying
that
$\varStatus_{p'}^{t} = \OUT^{*}$.
The assertion is established by setting
$\tClean = t^{2} + 4$.
\end{proof}

Let
$\tClean \in \Integers_{> 0}$
be the time promised in Lemma~\ref{lem:mis:analysis:clean-time}.
Our goal in the remainder of this section is to prove that starting from time
$\tClean$, the algorithm stabilizes within
$O (\log^{2} n)$
rounds w.h.p.
To this end, we introduce the following additional definitions.

\begin{definition*}[%
$\IN^{t}$,
$\OUT^{t}$,
$G^{t} = (V^{t}, E^{t})$,
$n^{t}$,
$m^{t}$,
$\Degree^{t}(v)$]
For
$t \geq \tClean$,
let
$\IN^{t} \subseteq V$
and
$\OUT^{t} \subseteq V$
be the sets of $\IN$-nodes and $\OUT$-nodes, respectively, at time $t$.
Let
$V^{t} = V - (\IN^{t} \cup \OUT^{t})$
be the set of undecided nodes at time $t$ and let
$G^{t} = (V^{t}, E^{t})$
be the subgraph of $G$ induced by $V^{t}$;
let
$n^{t} = |V^{t}|$
and
$m^{t} = |E^{t}|$
be the number of nodes and edges, respectively, in $G^{t}$.
Let
$\Degree^{t}(\cdot) = \Degree_{G^{t}}(\cdot)$
and extend the scope of this operator to all nodes in $V$ by defining
$\Degree^{t}(v) = 0$
for every node
$v \in V - V^{t}$.
\end{definition*}

The definition of clean times ensures that the following four properties hold
for every node
$v \in V$
and time
$t \geq \tClean$:
(1)
if
$v \in \IN^{t}$,
then
$u \notin \IN^{t}$
for every
$u \in \Neighbors(v)$;
(2)
if
$v \in \OUT^{t}$,
then there exists
$u \in \Neighbors(v)$
such that
$u \in \IN^{t}$;
(3)
$\IN^{t} \subseteq \IN^{t + 1}$;
and
(4)
$\OUT^{t} \subseteq \OUT^{t + 1}$.
Therefore, if all nodes are decided at time
$t \geq \tClean$,
that is,
$V^{t} = \emptyset$,
then the algorithm has stabilized to a legal output assignment by time $t$.
The rest of the analysis is dedicated to proving that 
$V^{\tClean + O (\log^{2} n)} = \emptyset$
w.h.p.
We start with the following observation that holds due to the update rule of
the reorientation bit variables.

\begin{observation}
\label{obs:mis:analysis:node-reorientation-bit}
For every node
$v \in V$,
port
$p \in \Ports(v)$,
and time
$t > \tClean$,
we have
(1)
if
$v \in \IN^{t}$,
then
$\varReorientationBit_{p}^{t} = 1$;
(2)
if
$v \in \OUT^{t}$,
then
$\varReorientationBit_{p}^{t} = 0$;
and
(3)
if
$v \in V^{t}$,
then
$\varReorientationBit_{p}^{t} = \varCoin_{p^{*}(v)}^{t - 1}$.
\end{observation}

We can now establish the following observation that allows us to focus on
$E^{t}$ instead of $V^{t}$.

\begin{observation}
\label{obs:mis:analysis:isolated-nodes}
For every time
$t > \tClean$
and node
$v \in V^{t}$,
if $v$ has no neighbors in $G^{t}$ (i.e., it is an isolated node), then
$v \notin V^{t + O (\log n)}$
w.h.p.
\end{observation}
\begin{proof}
If there exists a node
$u \in \Neighbors_{G}(v)$
such that 
$u \in \IN^{t}$,
then $v$ is guaranteed to become an $\OUT$ node by time
$t + 2$.
So, assume that
$\Neighbors_{G}(v) \subseteq \OUT^{t}$
and consider some edge
$e = (u, v)$
incident on $v$.
Observation~\ref{obs:mis:analysis:node-reorientation-bit} ensures that
$\varReorientationBit_{(u, e)}^{t'} = 0$
for every time
$t' \geq t$.
Moreover, if $v$ is still undecided at time $t'$, then
$\varReorientationBit_{(v, e)}^{t'}
=
\varCoin_{p^{*}(v)}^{t' - 1}
\sim
\UnifDist(\{ 0, 1 \})$.
Therefore, edge $e$ is guaranteed to be oriented toward $v$ no later than time
$t + O (\log n)$
w.h.p.\ and remain oriented toward $v$ subsequently.
The assertion follows by a union bound over all edges incident on $v$ as once
all these edges are oriented toward $v$, node $v$ becomes an $\IN$-node.
\end{proof}

Owing to Observation~\ref{obs:mis:analysis:isolated-nodes}, our goal is to
prove that
$m^{\tClean + O (\log^{2} n)} = 0$
w.h.p.
To this end, we prove that there exist universal constants
$\alpha \in \Integers_{> 0}$
and
$c > 0$
such that
\begin{equation}
\label{eq:mis:analysis:main}
\Ex \left( m^{t + \lceil \log n \rceil + \alpha} \mid m^{t} \right)
\leq
(1 - c) m^{t}
\qquad\text{a.s.}
\end{equation}
for every
$t > \tClean$.
Theorem~\ref{thm:mis} then follows as a corollary of
Lemma~\ref{lem:prelim:probabilistic-progress}.

The remainder of this section is dedicated to establishing
\eqref{eq:mis:analysis:main}.
As a first step toward this goal, we recall a classic combinatorial lemma of
Alon, Babai, and Itai~\cite{AlonBI1986mis} for which we need the following
definition.

\begin{definition*}[%
good nodes]
Consider an undirected graph $H$.
A node
$v \in V_{H}$
of degree
$d = \Degree_{H}(v)$
is said to be \emph{good} in $H$ if
$\left| \left\{
u \in \Neighbors_{H}(v) : \Degree_{H}(u) \leq d
\right\} \right|
\geq
d / 3$.
\end{definition*}

\begin{lemma}[\cite{AlonBI1986mis}]
\label{lem:mis:analysis:good-nodes}
Let $H$ be an undirected graph and let
$\Gamma \subseteq V_{H}$
be the set of good nodes in $H$.
Then,
$|\{ e \in E_{H} : e \cap \Gamma \neq \emptyset \}|
\geq
|E_{H}| / 2$.
\end{lemma}

We establish \eqref{eq:mis:analysis:main} by combining
Lemma~\ref{lem:mis:analysis:good-nodes} and the following lemma.

\begin{lemma}
\label{lem:mis:analysis:good-degree-decreases}
Fix some time
$t > \tClean$
and the configuration
$C^{t}$
and consider a node
$v \in V^{t}$
of degree
$d = \Degree^{t}(v) > 0$
that is good in the graph $G^{t}$.
There exists a universal constant
$c > 0$
such that
\[
\Ex \left( \Degree^{t + \lceil \log d \rceil + 4}(v) \right)
\leq
(1 - c) d
\, .
\]
\end{lemma}
\begin{proof}
Let
$t' = t + \lceil \log d \rceil + 1$.
Given a node
$u \in V^{t}$,
we say that $u$ is \emph{out-destined} at time $t'$ if there exists a node
$w \in \Neighbors_{G}(u) \cap \IN^{t'}$.
Let $D(u)$ be the event that $u$ is out-destined at time $t'$ and recall that
the algorithm is designed so that $D(u)$ implies that $u$ is an $\OUT$-node at
time
$t' + 2 = t + \lceil \log d \rceil + 3$
(at the latest).
Node $u$ is regarded as \emph{weak} (with respect to configuration $C^{t}$) if
$\Pr (D(u)) \geq 1 / 8$;
otherwise, node $u$ is regarded as \emph{strong}.

Let $v$ be the node from the lemma's statement and let
$\Neighbors^{\leq}_{G^{t}}(v)
=
\{ u \in \Neighbors_{G^{t}}(v) : \Degree^{t}(u) \leq d \}$,
recalling that
$|\Neighbors^{\leq}_{G^{t}}(v)| \geq d / 3$
as $v$ is good in $G^{t}$;
let
$s_{1}, \dots, s_{k}$
be the strong nodes in
$\Neighbors^{\leq}_{G^{t}}(v)$.
If more than half of the nodes in
$\Neighbors^{\leq}_{G^{t}}(v)$
are weak, then
$\Ex (\Degree^{t + \lceil \log d \rceil + 4}(v))
\leq
\Ex (\Degree^{t + \lceil \log d \rceil + 3}(v))
<
d \left( 1 - \frac{1}{42} \right)$
by the linearity of expectation, thus establishing the assertion.
So, assume in what follows that at least half of the nodes in
$\Neighbors^{\leq}_{G^{t}}(v)$
are strong, i.e.,
$k
\geq
|\Neighbors^{\leq}_{G^{t}}(v)| / 2
\geq
d / 6$.

We say that a node
$u \in V^{t}$
is \emph{lucky} if
$\varCoin_{p^{*}(u)}^{\tau - 1} = 1$
for all
$t \leq \tau < t'$
and denote this event by
$L(u)$;
clearly, the events in
$\{ L(u) \}_{u \in V^{t}}$
are mutually independent and
\[
\frac{1}{2 d}
\leq
\Pr (L(u))
<
\frac{1}{4 d}
\]
as
$t' - t = \lceil \log d \rceil + 1$.
We can now introduce the following events for each
$1 \leq i \leq k$:
\[
A_{i} = L(s_{i}) \land \bigwedge_{1 \leq j \leq i - 1} \neg L(s_{j})
\, ,
\quad
B_{i}
=
\bigwedge_{u \in \Neighbors_{G^{t}}(s_{i}) - \{ s_{1}, \dots, s_{i - 1} \}}
\neg L(u)
\, ,
\quad\text{and}\quad
D_{i} = D(s_{i})
\, .
\]

The key observation now is that
$A_{i} \land B_{i} \land \neg D_{i}$
implies that all edges (in $E$) incident on $s_{i}$ are oriented toward
$s_{i}$ at time $t'$.
This in turn implies, by the update rule of the status variables, that
$s_{i} \in \IN^{t' + 1}$
with probability at least
$1 / 2$.
Recalling that
$s_{i} \in \IN^{t' + 1}
\Longrightarrow
v \in \OUT^{t' + 3}$,
we conclude that
$A_{i} \land B_{i} \land \neg D_{i}$
implies that
$\Degree^{t + \lceil \log d \rceil + 4}(v) = \Degree^{t' + 3}(v) = 0$
with probability at least
$1 / 2$.
The assertion is established by proving that
\begin{equation}
\label{eq:mis:analysis:good-degree-decreases:goal-1}
\Pr \left(
\bigvee_{1 \leq i \leq k} A_{i} \land B_{i} \land \neg D_{i}
\right)
>
\frac{1}{192}
\, .
\end{equation}

\sloppy
Since the events in
$\{ A_{i} \}_{1 \leq i \leq k}$
are pairwise disjoint, it follows that the events in
\mbox{$\{ A_{i} \land B_{i} \land \neg D_{i} \}_{1 \leq i \leq k}$}
are pairwise disjoint, hence
\[
\Pr \left(
\bigvee_{1 \leq i \leq k} A_{i} \land B_{i} \land \neg D_{i}
\right)
=
\sum_{1 \leq i \leq k} \Pr (A_{i} \land B_{i} \land \neg D_{i})
\, .
\]
We shall establish
\eqref{eq:mis:analysis:good-degree-decreases:goal-1}
by proving that
\begin{equation}
\label{eq:mis:analysis:good-degree-decreases:goal-2}
\Pr (A_{i} \land B_{i} \land \neg D_{i})
>
\frac{1}{32 d}
\geq
\frac{1}{192 k}
\end{equation}
for each
$1 \leq i \leq k$.
\par\fussy

To this end, consider some
$1 \leq i \leq k$
and develop
\begin{equation}
\label{eq:mis:analysis:good-degree-decreases:bound-A}
\Pr (A_{i})
=
\Pr (L(s_{i}))
\cdot
\left( 1 - \Pr \left( \bigvee_{1 \leq j \leq i - 1} L(s_{j}) \right) \right)
>
\frac{1}{4 d} \cdot \left( 1 - (i - 1) \cdot \frac{1}{2 d} \right)
>
\frac{1}{8 d}
\, ,
\end{equation}
where
the second transition is by the union bound
and
the last transition holds as
$i - 1 < k \leq d$,
and
\begin{equation}
\label{eq:mis:analysis:good-degree-decreases:bound-B}
\Pr (B_{i})
=
1 - \Pr \left(
\bigvee_{u \in \Neighbors_{G^{t}}(s_{i}) - \{ s_{1}, \dots, s_{i - 1} \}} L (u)
\right)
\geq
1 - \Degree^{t}(s_{i}) \cdot \frac{1}{2 d}
\geq
\frac{1}{2}
\, ,
\end{equation}
where
the second transition is by the union bound
and
the last transition holds as
$\Degree^{t}(s_{i}) \leq d$.
We further develop
\begin{align}
\Pr (D_{i} \mid A_{i} \land B_{i})
= \, &
\frac{\Pr (D_{i} \land A_{i} \land B_{i})}{\Pr (A_{i} \land B_{i})}
=
\frac{\Pr (D_{i} \land A_{i} \land B_{i})}%
{\Pr (A_{i}) \cdot \Pr (B_{i})}
<
16 d \cdot \Pr (D_{i} \land A_{i} \land B_{i})
\nonumber
\\
\leq \, &
16 d \cdot \Pr (D_{i} \land L(s_{i}))
=
16 d \cdot \Pr (D_{i} \mid L(s_{i})) \cdot \Pr (L(s_{i}))
\nonumber
\\
\leq \, &
16 d \cdot \Pr (D_{i}) \cdot \Pr (L(s_{i}))
<
16 d \cdot \frac{1}{8} \cdot \frac{1}{4 d}
=
\frac{1}{2}
\label{eq:mis:analysis:good-degree-decreases:bound-D-cond}
\, ,
\end{align}
where
the second transition holds as events $A_{i}$ and $B_{i}$ are independent,
the third transition follows from
\eqref{eq:mis:analysis:good-degree-decreases:bound-A} and
\eqref{eq:mis:analysis:good-degree-decreases:bound-B},
the sixth transition holds by the definition of event $L(s_{i})$,
and the penultimate transition holds by recalling that $s_{i}$ is strong.
%YE: the sixth transition requires a more elaborated explanation.
Inequality~%
\eqref{eq:mis:analysis:good-degree-decreases:goal-2}
is now established by developing
\begin{align*}
\Pr (A_{i} \land B_{i} \land \neg D_{i})
= \, &
\Pr (\neg D_{i} \mid A_{i} \land B_{i}) \cdot \Pr (A_{i} \land B_{i})
\\
= \, &
\Pr (\neg D_{i} \mid A_{i} \land B_{i})
\cdot
\Pr (A_{i}) \cdot \Pr(B_{i})
>
\frac{1}{2} \cdot \frac{1}{8 d} \cdot \frac{1}{2}
=
\frac{1}{32 d}
\, ,
\end{align*}
where
the second transition holds as events $A_{i}$ and $B_{i}$ are independent
and
the penultimate transition follows from
\eqref{eq:mis:analysis:good-degree-decreases:bound-D-cond},
\eqref{eq:mis:analysis:good-degree-decreases:bound-A},
and
\eqref{eq:mis:analysis:good-degree-decreases:bound-B}.
\end{proof}

%%%%%%%%%%%%%%%%%%%%%%%%%%%%%%%%%%%%%%%%%%%%%%%%%%%%%%%%%%%%%%%%%%%%%%%%%%%%%%
\section{Maximal Matching}
\label{sec:mm}
%%%%%%%%%%%%%%%%%%%%%%%%%%%%%%%%%%%%%%%%%%%%%%%%%%%%%%%%%%%%%%%%%%%%%%%%%%%%%%
An edge subset
$M \subseteq E$
is a \emph{matching} if
$e \cap e' = \emptyset$
for every distinct edges
$e, e' \in M$.
A \emph{maximal matching (MM)} is a matching
$M \subseteq E$
that is maximal in the sense that $M'$ is not a matching for
every
$M \subset M' \subseteq E$.
In the MM problem, the goal is to construct a MM
$M \subseteq E$.
Under the UP model, this is translated to using
$\{ \MATCHED, \UNMATCHED, \bot \}$
as the output label set so that an edge
$e = \{ v_{1}, v_{2} \} \in E$
is determined to be included in (resp., excluded from) the constructed MM $M$
in a configuration $C$ if
$\omega(C((v_{i}, e))) = \MATCHED$
(resp.,
$\omega(C((v_{i}, e))) = \UNMATCHED$)
for each
$i \in \{ 1, 2 \}$;
edge $e$ is regarded as \emph{undecided} in $C$ if it is neither included in,
nor excluded from, $M$.
A node
$v \in V$
is regarded as \emph{undecided} in $C$ if at least one of its incident edges
is undecided;
otherwise, $v$ is regarded as \emph{decided} in $C$.

\begin{theorem}
\label{thm:mm}
There exists a self-stabilizing UP algorithm that solves the MM problem and
stabilizes in
$O(\log^{5} n)$
time w.h.p.
\end{theorem}

The algorithm promised in Theorem~\ref{thm:mm} is presented in
Section~\ref{sec:mm:algorithm}, first at a high level and then, from the
perspective of the individual ports.
Section~\ref{sec:mm:analysis} is then dedicated to proving the correctness of
our algorithm and bounding its stabilization time, thus establishing
Theorem~\ref{thm:mm}.

%%%%%%%%%%%%%%%%%%%%%%%%%%%%%%%%%%%%%%%
\subsection{The MM Algorithm}
\label{sec:mm:algorithm}
%%%%%%%%%%%%%%%%%%%%%%%%%%%%%%%%%%%%%%%

%%%%%%%%%%%%%%%%%%%%%%%%%%%%%%%%%%%%%%%
\subsubsection*{High-Level Description}
%%%%%%%%%%%%%%%%%%%%%%%%%%%%%%%%%%%%%%%
Our MM algorithm is inspired by a classic MM algorithm of Israeli and Itai \cite{IsraelI1986mm} (in fact, by an adaptation thereof, presented in \cite{BittonEIK2024mm}).
The algorithm of \cite{IsraelI1986mm} divides the execution into \emph{phases},
where each phase lasts for a fixed (constant) number of rounds and is
performed in synchrony by all undecided nodes.
Each phase of an undecided node
$v \in V$
has a \emph{mode} picked by $v$ uniformly at
random from
$\{ \Proposing, \Receiving \}$.
In a $\Proposing$-phase, $v$ \emph{proposes} to an undecided neighbor picked
uniformly at random.
In a $\Receiving$-phase, if $v$ receives at least one proposal, then $v$
selects exactly one of the proposing neighbors
$u \in \Neighbors(u)$
(arbitrarily) and accepts $u$'s proposal, which results in adding the edge
$\{ u, v \}$
to the constructed matching and turning both $v$ and $u$ into decided nodes.

When trying to implement this general scheme as a self-stabilizing UP
algorithm, we encounter two obstacles:
First, in the realm of self-stabilizing algorithms, the phases of the
individual nodes do not necessarily run in synchrony.\footnote{%
Bitton, Emek, Izumi, and Kutten \cite{BittonEIK2024mm} deal with this obstacle by introducing a
general technique called probabilistic phase synchronization.
This technique however assumes a fixed phase length, a property that does not
hold in our case.}
Second, under the UP model, picking one neighbor (uniformly at random or
otherwise) is equivalent to picking one port;
this requires symmetry breaking among the sibling ports which generally
involves a stochastic process whose length is a random variable, thus we
cannot hope for fixed length phases.

Our UP MM algorithm adheres to the aforementioned phase structure, picking each phase of an undecided node
$v \in V$
to be either $\Proposing$ or $\Receiving$ uniformly at random, and to overcome the aforementioned obstacles, we divide the phase into two \emph{epochs}.
In a $\Proposing$-phase, the 1st epoch is dedicated to picking an undecided
neighbor
$u \in \Neighbors(v)$
uniformly at random;
the 2nd epoch is then dedicated to proposing to $u$ via port
$(v, \{ u, v \})$.
In a $\Receiving$-phase, the 1st epoch has no particular role and the 2nd
epoch is dedicated to picking the proposal of
exactly one neighbor
$u \in \Neighbors(v)$
(assuming that $v$ receives at least one proposal);
the proposal is then accepted via the port
$(v, \{ u, v \})$.

The length of each epoch is determined by means of a \emph{(port) tournament}:
initially, all ports in $\Ports(v)$ are \emph{candidates} and in every
round, each candidate port \emph{retires} --- i.e., stops being a candidate ---
independently with probability
$1 / 2$;
the tournament ends once all ports have retired.
This repeats itself in every epoch of every phase so that the epoch starts with all ports being candidates and ends once all ports have retired.

In a $\Proposing$-phase, the tournament structure of the 1st epoch
facilitates picking one undecided neighbor
$u \in \Neighbors(v)$
uniformly at random:
$u$ is picked if and only if port
$p = (v, \{ u, v \})$
remains the unique last candidate in $v$'s (1st epoch) tournament
among all ports
$(v, \{ u', v \})$
such that $u'$ is undecided.
(Notice that if all surviving candidates retire in the same round, then the last candidate is not unique, in which case, node $v$ fails to match in the current phase;
this is accounted for in the analysis, see Section~\ref{sec:mm:analysis}.)
Following that, during the 2nd epoch of the phase, $v$ attempts to propose
to $u$, however, this happens with one crucial condition:
the 2nd epoch of $v$ should be perfectly synchronized with that of $u$,
that is, the corresponding tournaments start at the same time.
Assuming that this condition is satisfied, the proposals of $v$ to $u$ are
carried out through port $p$ and continue as long as $p$ is a candidate in
$v$'s (2nd epoch) tournament.

In a $\Receiving$-phase, the 1st epoch has no particular role (although it still lasts for exactly one tournament), whereas the role of the 2nd epoch is to pick one out of possibly multiple incoming proposals.
To this end, we exploit, once again, the tournament structure:
Consider the set
$Q \subseteq \Neighbors(v)$
of undecided neighbors of $v$ that
(1)
start a 2nd epoch tournament in synchrony with the 2nd epoch tournament
of $v$;
and
(2)
propose to $v$ when the tournament begins.
Recall that the nodes
$u \in Q$
keep proposing to $v$ as long as port
$(u, \{ u, v \})$
is a candidate in the corresponding 2nd epoch tournament of $u$.
Node $v$ accepts the proposal of the node
$u \in Q$
that remains the unique last proposing node.
(As before, the last proposing node in $Q$ may not be unique, in which case,
node $v$ fails to match in the current phase;
this is accounted for in the analysis, see Section~\ref{sec:mm:analysis}.)

\begin{remark*}
On the face of it, the structure of the algorithm may seem overly complicated:
As mentioned earlier, the 1st epoch of a $\Receiving$-phase is ``redundant''.
Moreover, it is not clear why the 2nd epoch of a $\Proposing$-phase does not
end as soon as the port through which the proposals are carried out retires.
Put differently, is it really necessary that each phase lasts for (exactly)
two i.i.d.\ tournaments?

The answer to this question is categorically positive:
Recall that a matching between a proposing node $u$ and a receiving node $v$
can be finalized only if the 2nd epochs of their corresponding ($\Proposing$-
and $\Receiving$-)phases are synchronized.
The analysis presented in Section~\ref{sec:mm:analysis} crucially depends on
the argument that such
$(u, v)$-synchronization
events occur sufficiently often (see
Section~\ref{sec:mm:analysis:tournaments}).
It turns out that the key for establishing this argument is to ensure that the
lengths of the phases of $u$ and $v$ are independent, a property that we
obtain by insisting that each phase lasts for two i.i.d.\ tournaments.
\end{remark*}

%%%%%%%%%%%%%%%%%%%%%%%%%%%%%%%%%%%%%%%
\subsubsection*{Port-Level Description}
%%%%%%%%%%%%%%%%%%%%%%%%%%%%%%%%%%%%%%%
The mechanism behind the tournaments at a node
$v \in V$
is controlled by a variable
$p.\varTournament \in \{ 0, 1 \}$,
referred to as the \emph{tournament} variable, that each port
$p \in \Ports(v)$
maintains.
This variable is updated in each round according to the following rule:
if
$\varTournament_{p} = 1$,
then
$\varTournament_{p} \sim \UnifDist(\{ 0, 1 \})$,
independently;
if
$\varTournament_{p'} = 0$
for all
$p' \in \Siblings(p) \cup \{ p \}$,
then
$\varTournament_{p} \gets 1$;
otherwise, the value of $\varTournament_{p}$ remains unchanged (formally,
$\varTournament_{p} \gets \varTournament_{p}$).
This means that each tournament ends in a round during which all ports
$p \in \Ports(v)$
raise the values of their tournament variables from
$\varTournament_{p} = 0$
to
$\varTournament_{p} = 1$;
the next round is regarded as the first round of the subsequent
tournament.\footnote{%
In the language of the aforementioned high-level description, a port
$p \in \Ports(v)$
is regarded as a candidate in a tournament as long as
$\varTournament_{p} = 1$.
}

For each node
$v \in V$,
the actions of the ports in $\Ports(v)$ are orchestrated by a leader port
$p^{*}(v) \in \Ports(v)$
selected by running the LLE mechanism in $v$ (see
Section~\ref{sec:prelim}).
Port
$p^{*}(v)$
maintains two designated variables that the other (``ordinary'') ports do not
maintain:
variable
$\varPhaseMode_{p^{*}(v)} \in \{ \Proposing, \Receiving \}$,
referred to as the \emph{phase mode} variable, that holds the mode of $v$'s
current phase;
and
variable
$\varEpoch_{p^{*}(v)} \in \{ 1, 2, 2_{0} \}$,
referred to as the \emph{epoch} variable, that records the epoch of $v$'s
current phase, where the symbol $2_{0}$ indicates the first round of the 2nd
epoch (as we shall see soon, it is important to distinguish this round from
the rest of the rounds in the 2nd epoch).
To avoid cumbersome notation, we subsequently denote these variables by
$\varPhaseMode_{v}$ and $\varEpoch_{v}$ instead of $\varPhaseMode_{p^{*}(v)}$
and $\varEpoch_{p^{*}(v)}$, respectively.

The epoch variable of node $v$ is updated in each round according to the
following rule:
\begin{itemize}

\item
Assign
$\varEpoch_{v} \gets 1$
if
(I)
$\varTournament_{p} = 0$
for all
$p \in \Ports(v)$;
and
(II)
$\varEpoch_{v} = 2$.

\item
Otherwise, assign
$\varEpoch_{v} \gets 2_{0}$
if
(I)
$\varTournament_{p} = 0$
for all
$p \in \Ports(v)$;
and
(II)
$\varEpoch_{v} = 1$.

\item
Otherwise, assign
$\varEpoch_{v} \gets 2$
if
$\varEpoch_{v} = 2_{0}$.

\item
Otherwise, the value of $\varEpoch_{v}$ remains unchanged (formally,
$\varEpoch_{v} \gets \varEpoch_{v}$).

\end{itemize}
The phase mode variable $\varPhaseMode_{v}$ is updated whenever the value of
$\varEpoch_{v}$ changes from
$\varEpoch_{v} = 2$
to
$\varEpoch_{v} = 1$
(i.e.,
$\varEpoch_{v} = 2$
and
$\varEpoch^{+}_{v} = 1$);
this update consists of picking
$\varPhaseMode_{v} \sim \UnifDist(\{ \Proposing, \Receiving \})$,
independently.

The output label associated with each port
$p \in \Ports$
is determined based on the value of the designated variable
$\varStatus_{p} \in \{ \MATCHED, \UNMATCHED, \UNMATCHED^{*}, \bot \}$,
referred to as the \emph{status} of $p$.
Specifically, the output label associated with $p$ is
$\MATCHED$, $\UNMATCHED$, or $\bot$
if
$\varStatus_{p} = \MATCHED$,
$\varStatus_{p} \in \{ \UNMATCHED, \UNMATCHED^{*} \}$,
or
$\varStatus_{p} = \bot$,
respectively.
The update rule of this variable depends on three additional variables and is
therefore deferred until those three variables are presented.

In a $\Proposing$-phase of node $v$, the proposal mechanism is handled with
the help of a variable
$\varProposal_{p} \in \{ \Standby, \Proposing, \Proposing_{0}, \bot \}$,
referred to as the \emph{proposal} variable, that each port
$p \in \Ports(v)$
maintains as long as
$\varStatus_{p} = \bot$.
The semantics of these values is as follows:
$\varProposal_{p} = \Standby$
indicates that $p$ is selected as the proposing port of the current
($\Proposing$-)phase;
and
$\varProposal_{p} \in \{ \Proposing, \Proposing_{0} \}$
indicates that $p$ is actively proposing to its counter-port, where
$\Proposing_{0}$ is reserved for the first round of the 2nd epoch.
Specifically, this variable is updated in each round according to the
following rule:
\begin{itemize}

\item
Assign
$\varProposal_{p} \gets \Standby$
if and only if one of the following two conditions is satisfied:
\begin{itemize}
  \item
  (I)
  $\varPhaseMode_{v} = \Proposing$;
  (II)
  $\varEpoch_{v} = 1$;
  (III)
  $\varTournament_{p} = 1$;
  and
  (IV)
  $\varTournament_{p'} = 0$
  for all
  $p' \in \Siblings(p)$
  such that
  $\varStatus_{p'} = \bot$.

  \item
  (I)
  $\varPhaseMode_{v} = \Proposing$;
  (II)
  $\varEpoch_{v} = 1$;
  (III)
  $\varProposal_{p} = \Standby$
  (IV)
  there exists a port
  $p' \in \Siblings_{p}$
  such that
  $\varTournament_{p'} = 1$. 

\end{itemize}

\item
Otherwise, assign
$\varProposal_{p} \gets \Proposing_{0}$
if
(I)
$\varPhaseMode_{v} = \Proposing$;
(II)
$\varEpoch_{v} = 1$;
(III)
$\varProposal_{p} = \Standby$;
and
(IV)
$\varTournament_{p'} = 0$
for all
$p' \in \Siblings(p) \cup \{ p \}$.

\item
Otherwise, assign
$\varProposal_{p} \gets \Proposing$
if
(I)
$\varPhaseMode_{v} = \Proposing$;
(II)
$\varEpoch_{v} \in \{ 2, 2_{0} \}$;
(III)
$\varProposal_{p} \in \{ \Proposing, \Proposing_{0} \}$;
(IV)
$\varProposal_{p'} = \bot$
for all
$p' \in \Siblings(p)$;
and
(V)
$\varTournament^{+}_{p} = 1$.

\item
Otherwise, assign
$\varProposal_{p} \gets \bot$.

\end{itemize}

In a $\Receiving$-phase of node $v$, the mechanism that determines which
proposal is accepted (if any) is handled with the help of a variable
$\varReceivingFlag_{p} \in \{ \True, \False \}$,
referred to as the receiving flag, that each port
$p \in \Ports(v)$
maintains as long as
$\varStatus_{p} = \bot$.
Port $p$ updates this variable in each round by setting
$\varReceivingFlag_{p} \gets \True$
if and only if one of the following two conditions is satisfied:
\begin{itemize}

\item
(I)
$\varPhaseMode_{v} = \Receiving$;
(II)
$\varEpoch_{v} = 2_{0}$;
and
(III)
$\varProposal_{\CounterPort{p}} = \Proposing_{0}$.

\item
(I)
$\varPhaseMode_{v} = \Receiving$;
(II)
$\varEpoch_{v} = 2$;
(III)
$\varReceivingFlag_{p} = \True$;
and
(IV)
$\varProposal_{\CounterPort{p}} = \Proposing$.

\end{itemize}

The actual proposal acceptance mechanism is handled with the help of a
variable
$\varAcceptFlag_{p} \in \{ \True, \False \}$,
referred to as the \emph{acceptance flag}, that each port
$p \in \Ports(v)$
maintains as long as
$\varStatus_{p} = \bot$.
Port $p$ updates this variable in each round by setting
$\varAcceptFlag_{p} \gets \True$
if and only if
(I)
$\varPhaseMode_{v} = \Receiving$;
(II)
$\varEpoch_{v} = 2$;
(III)
$\varReceivingFlag_{p} = \True$;
(IV)
$\varReceivingFlag_{p'} = \False$
for all
$p' \in \Siblings(p)$
such that
$\varStatus_{p'} = \bot$;
(V)
there exists a port
$p' \in \Siblings(p) \cup \{ p \}$
such that
$\varTournament_{p'} = 1$;
and
(VI)
$\varProposal_{\CounterPort{p}} = \Proposing$.

We are now ready to present the update rule of the status variable
$\varStatus_{p}$ maintained by each port
$p \in \Ports(v)$:
\begin{itemize}

\item
Assign
$\varStatus_{p} \gets \MATCHED$
if
$\MATCHED \notin \{ \varStatus_{p'} : p' \in \Siblings(p) \}$
and
at least one of the following conditions is satisfied:
  \begin{itemize}

  \item
(I)
$\varPhaseMode_{v} = \Receiving$;
(II)
$\varEpoch_{v} = 2$;
and
(III)
$\varAcceptFlag_{p} = \True$.

  \item
(I)
$\varPhaseMode_{v} = \Proposing$;
(II)
$\varEpoch_{v} = 2$;
and
(III)
$\varAcceptFlag_{\CounterPort{p}} = \True$.

  \item
$\varStatus_{p} = \varStatus_{\CounterPort{p}} = \MATCHED$.

  \end{itemize}

\item
Otherwise, assign
$\varStatus_{p} \gets \UNMATCHED^{*}$
if there exists a port
$p' \in \Siblings(p)$
such that
$\varStatus_{p'} = \MATCHED$.

\item
Otherwise, assign
$\varStatus_{p} \gets \UNMATCHED$
if
$\varStatus_{\CounterPort{p}} = \UNMATCHED^{*}$.

\item
Otherwise, assign
$\varStatus_{p} \gets \bot$.

\end{itemize}

%%%%%%%%%%%%%%%%%%%%%%%%%%%%%%%%%%%%%%%
\subsection{Analysis}
\label{sec:mm:analysis}
%%%%%%%%%%%%%%%%%%%%%%%%%%%%%%%%%%%%%%%
Throughout the analysis, we fix an arbitrary initial configuration $C^{0}$ and
consider the (random) execution
$\eta = \{ C^{t} \}_{t \geq 0}$
of the MM algorithm starting from $C^{0}$.
A time
$t \in \Integers_{> 0}$
is said to be \emph{clean} with respect to $\eta$ if the following conditions
hold for every node
$v \in V$
and port
$p \in \Ports(v)$:
\begin{enumerate}

\item
\label{mm:analysis:clean-condition:leaders}
\label{mm:analysis:clean-condition:first}
Node $v$ admits a unique leader port in
$C^{t}$.

\item
If
$\varProposal_{p}^{t} \neq \bot$,
then
(I)
$\varPhaseMode_{v}^{t} = \Proposing$;
and
(II)
$\varProposal_{p'}^{t} = \bot$
for all
$p' \in \Siblings(p)$.

\item
If
$\varProposal_{p}^{t} = \Standby$,
then
$\varEpoch_{v}^{t} = 1$

\item
If
$\varProposal_{p}^{t} = \Proposing^{0}$,
then
$\varEpoch_{v}^{t} = 2_{0}$.

\item
If
$\varProposal_{p}^{t} = \Proposing$,
then
(I)
$\varEpoch_{v}^{t} = 2$;
and
(II)
$\varTournament_{p}^{t} = 1$.

\item
If
$\varReceivingFlag_{p} = \True$,
then
(I)
$\varPhaseMode_{v} = \Receiving$;
and
(II)
$\varEpoch_{v} = 2$.

\item
If
$\varAcceptFlag_{p}^{t} = \True$,
then
(I)
$\varReceivingFlag_{p} = \True$;
and
(II)
$\varReceivingFlag_{p'}^{t} = \False$
for all
$p' \in \Siblings(p)$
such that
$\varStatus_{p'} = \bot$.

\item
If
$\varStatus_{p}^{t} = \MATCHED$,
then
(I)
$\varStatus_{\CounterPort{p}}^{t} = \MATCHED$;
and
(II)
$\varStatus_{p'}^{t} \neq \MATCHED$
for all
$p' \in \Siblings(p)$.

\item
If
$\varStatus_{p}^{t} = \UNMATCHED^{*}$,
then there exists
$p' \in \Siblings(p)$
such that
$\varStatus_{p'}^{t} = \MATCHED$.

\item
If
$\varStatus_{p}^{t} = \UNMATCHED$,
then
$\varStatus_{\CounterPort{p}}^{t} = \UNMATCHED^{*}$.

\end{enumerate}
Lemma~\ref{lem:prelim:up-model:local-port-election} guarantees that there
exists a time
$t^{0} \in \Integers_{> 0}$
such that every node admits a (fixed) leader port at all times
$t \geq t^{0}$
and
$t^{0} \leq O (\log^{2} n)$
w.h.p.;
condition hereafter on this event.
The following observation is obtained by exhaustive case analysis.

\begin{observation}
\label{obs:mm:analysis:clean}
All times
$t > t_{0} + 3$
are clean with respect to $\eta$.
\end{observation}

Let
$\tClean = t_{0} + 3$.
Our goal in the remainder of this section is to prove that starting from time
$\tClean$, the algorithm stabilizes within
$O (\log^{5} n)$
rounds w.h.p.
To this end, we introduce the following additional definitions.

\begin{definition*}[%
$\IN^{t}$,
$\OUT^{t}$,
$G^{t} = (V^{t}, E^{t})$,
$n^{t}$,
$m^{t}$,
$\Degree^{t}(v)$]
For
$t \geq \tClean$,
let
$\IN^{t} \subseteq E$
(resp.,
$\OUT^{t} \subseteq E$)
be the set of edges
$e = \{ v_{1}, v_{2} \}$
such that
$\varStatus_{v_{i}, e}^{t} = \MATCHED$
(resp.,
$\varStatus_{v_{i}, e}^{t} \in \{ \UNMATCHED, \UNMATCHED^{*} \}$)
for each
$i \in \{ 1, 2 \}$.
Let
$E^{t} = E - (\IN^{t} \cup \OUT^{t})$
be the set of undecided edges at time $t$ and let
$G^{t} = (V^{t}, E^{t})$
be the subgraph of $G$ induced by $E^{t}$;
let
$n^{t} = |V^{t}|$
and
$m^{t} = |E^{t}|$
be the number of nodes and edges, respectively, in $G^{t}$.
Let
$\Degree^{t}(\cdot) = \Degree_{G^{t}}(\cdot)$
and extend the scope of this operator to all nodes in $V$ by defining
$\Degree^{t}(v) = 0$
for every node
$v \in V - V^{t}$.
\end{definition*}

The definition of clean times ensures that the following four properties hold
for every edge
$e \in E$
and time
$t \geq 0$:
(1)
if
$e \in \IN^{t}$,
then
$e' \notin \IN^{t}$
for every
$e' \in E$
such that
$e \cap e' \neq \emptyset$;
(2)
if
$e \in \OUT^{t}$,
then there exists
$e' \in E$
such that
$e \cap e' \neq \emptyset$
and
$e' \in \IN^{t}$;
(3)
$\IN^{t} \subseteq \IN^{t + 1}$;
and
(4)
$\OUT^{t} \subseteq \OUT^{t + 1}$.
Therefore, if all edges are decided at time
$t \geq \tClean$,
that is,
$E^{t} = \emptyset$,
then the algorithm has stabilized to a legal output assignment by time $t$.
The rest of the analysis is dedicated to proving that 
$E^{\tClean + O (\log^{5} n)} = \emptyset$,
or equivalently
$m^{\tClean + O (\log^{5} n)} = 0$,
w.h.p.
To this end, we prove that there exist universal constants
$\alpha \in \Integers_{> 0}$
and
$c > 0$
such that
\begin{equation}
\label{eq:mm:analysis:main}
\Ex \left( m^{t + \alpha \lceil \log^{3} n \rceil} \mid m^{t} \right)
\leq
\left( 1 - \frac{c}{\log n} \right) m^{t}
\qquad\text{a.s.}
\end{equation}
for every
$t > \tClean$.
Theorem~\ref{thm:mm} follows as a corollary of
Lemma~\ref{lem:prelim:probabilistic-progress}.

Recall the definition of good nodes from Section~\ref{sec:mis:analysis} and
the lemma of \cite{AlonBI1986mis} (Lemma~\ref{lem:mis:analysis:good-nodes})
ensuring that in any graph, at least half of the edges are incident on good
nodes.
Similarly to the proof structure in Section~\ref{sec:mis:analysis}, we
establish \eqref{eq:mm:analysis:main} by combining
Lemma~\ref{lem:mis:analysis:good-nodes} and the following lemma.

\begin{lemma}
\label{lem:mm:analysis:good-degree-decreases}
Fix some time
$t > \tClean$
and the configuration
$C^{t}$
and consider a node
$v \in V^{t}$
of degree
$d = \Degree^{t}(v) > 0$
that is good in the graph $G^{t}$.
There exist universal constants
$\alpha \in \Integers_{> 0}$
and
$c > 0$
such that
\[
\Ex \left( \Degree^{t + \alpha \lceil \log^{3} n \rceil}(v) \right)
\leq
\left( 1 - \frac{c}{\log n} \right) d
\, .
\]
\end{lemma}

Before we can prove
Lemma~\ref{lem:mm:analysis:good-degree-decreases}, we have
to introduce some additional definitions.

\begin{definition*}[%
geometric tournament,
length,
$Q$-winner,
repeated geometric tournament]
Fix some integer
$\kappa \in \Integers_{> 0}$
and consider $\kappa$ independent
$1 + \GeomDist{1 / 2}$
random variables
$X_{1}, \dots, X_{\kappa}$
(that is, each $X_{i}$ is defined over the integers
$x \geq 2$
so that
$\Pr (X_{i} = x) = 2^{-(x - 1)}$).
A \emph{geometric tournament} $\T$ consists of sampling
$X_{1}, \dots, X_{\kappa}$,
defining the \emph{length} of $\T$ to be
$L = L(\T) = \max_{i \in [\kappa]} X_{i}$.
Given a non-empty subset
$Q \subseteq \{ X_{1}, \dots, X_{\kappa} \}$,
we say that the random variable $X_{i}$ \emph{wins} the geometric tournament
$\T$ among $Q$, or, alternatively, that $X_{i}$ is the \emph{$Q$-winner} of
$\T$, if
$X_{i} > X_{i'}$
for all
$X_{i'} \in Q - \{ X_{i} \}$.
A \emph{repeated geometric tournament} is a stochastic process
$S = \{ S_{j} \}_{j = 0}^{\infty}$
over
$\Integers_{\geq 0}$
defined by sampling countably many independent geometric tournaments
$\T_{1}, \T_{2}, \dots$
and setting
$S_{j}
=
\sum_{h = 1}^{j} L(\T_{h})$.
\end{definition*}

The connection to the algorithm reveals itself by observing that each (port)
tournament of a node
$v \in V$
in the algorithm is a geometric tournament over the
$\{ 0, 1 \}$-coin
tosses of the $\varTournament$ variables associated with the $\Degree_{G}(v)$
ports in
$\Ports(v)$;
the notions of length and $Q$-winner translate accordingly.
Moreover, the sequence of starting times of $v$'s tournaments is a repeated
geometric tournament over $\Degree_{G}(v)$ random variables, where each phase
includes two consecutive tournaments.
The proof of
Lemma~\ref{lem:mm:analysis:good-degree-decreases} relies
on a careful analysis of the geometric tournaments;
this analysis is carried out in Section~\ref{sec:mm:analysis:tournaments},
culminating in Lemmas
\ref{lem:mm:analysis:geometric-tournament-length-upper-bound},
\ref{lem:mm:analysis:geometric-tournament-exact-length-and-winner},
\ref{lem:mm:analysis:two-geometric-tournaments},
and
\ref{lem:mm:analysis:repeated-geometric-tournament-mixing}.

\begin{lemma}
\label{lem:mm:analysis:geometric-tournament-length-upper-bound}
Consider an integer
$\kappa \in \Integers_{> 0}$,
$\kappa \leq n$,
and a geometric tournament $\T$ over the random variables
$X_{1}, \dots, X_{\kappa}$.
There exists an integer
$\nu = \nu(n) \in \Integers_{> 0}$
of order
$\nu = \Theta (\log n)$
such that the length of $\T$ is at most $\nu$
w.h.p.
\end{lemma}

\begin{lemma}
\label{lem:mm:analysis:geometric-tournament-exact-length-and-winner}
Consider an integer
$\kappa \in \Integers_{> 0}$
and a geometric tournament $\T$ over the random variables
$X_{1}, \dots, X_{\kappa}$.
There exists a universal constant
$c > 0$
and an integer
$\lambda = \lambda(\kappa) \in \Integers_{> 0}$
of order
$\lambda = \Theta (\log \kappa)$
such that for every non-empty subset
$Q \subseteq \{ X_{1}, \dots, X_{\kappa} \}$,
the following two conditions hold (simultaneously) with probability at least
$c$:
(I)
the length of $\T$ is $\lambda$;
and
(II)
$\T$ has a $Q$-winner.
Moreover, conditioned on the event that $\T$ has a $Q$-winner, each
$X_{i} \in Q$
is the winner with probability
$1 / |Q|$.
\end{lemma}

\begin{lemma}
\label{lem:mm:analysis:two-geometric-tournaments}
Consider two integers
$\kappa, \kappa' \in \Integers_{> 0}$,
$\kappa \leq \kappa'$,
and geometric tournaments $\T$ and $\T'$ over the random variables
$X_{1}, \dots, X_{\kappa}$
and
$X'_{1}, \dots, X'_{\kappa'}$,
respectively.
For every constant
$\alpha \in \Integers_{\geq 0}$,
there exists a constant
$c = c(\alpha) > 0$
such that
for every non-empty subset
$Q \subseteq \{ X_{1}, \dots, X_{\kappa} \}$,
the following two conditions hold (simultaneously) with probability at least
$c$:
(I)
$L(\T') \geq L(\T) + \alpha$;
and
(II)
$\T$ has a $Q$-winner.
\end{lemma}

\begin{lemma}
\label{lem:mm:analysis:repeated-geometric-tournament-mixing}
Consider an integer
$\kappa \in \Integers_{> 0}$
and a repeated geometric tournament
$S = \{ S_{j} \}_{j = 0}^{\infty}$
over the random variables
$X_{1}, \dots, X_{\kappa}$.
There exist an integer
$\mu = \mu(\kappa) \in \Integers_{> 0}$
of order
$\mu = \Theta (\log^{3} \kappa)$
and a universal constant
$c > 0$
such that for every integers $s$ and 
$j \geq 0$,
and for every integer
$z \geq \mu$,
conditioned on
$S_{j} = s$,
the probability that
$S_{j + h} = s + z$
for some even (resp., odd) integer
$h > 0$
is at least
$\frac{c}{\log \kappa}$.
\end{lemma}

We are now ready to establish
Lemma~\ref{lem:mm:analysis:good-degree-decreases}.

\begin{proof}[Proof of
Lemma~\ref{lem:mm:analysis:good-degree-decreases}]
Let $v$ be the node from the lemma's statement and let
$\{ u_{1}, \dots, u_{k} \}
=
\{ u \in \Neighbors_{G^{t}}(v) : \Degree^{t}(u) \leq d \}$,
recalling that
$k \geq d / 3$
as $v$ is good in $G^{t}$.
For
$i \in [k]$,
let
$d^{0}_{i} = \Degree_{G}(u_{i})$
and
$d_{i} = \Degree^{t}(u_{i})$.
Let
$\lambda_{i} = \Theta (\log d^{0}_{i})$
be the integer
$\lambda = \lambda(d^{0}_{i})$
promised in
Lemma~\ref{lem:mm:analysis:geometric-tournament-exact-length-and-winner}.
Let
$Q_{i} = \{ (u_{i}, \{ u_{i}, x \}) : x \in V_{t} \}$
and recall that
$|Q_{i}| = d_{i}$.

Let
$\mu = \mu(n) = \Theta (\log^{3} n)$
be the integer promised in
Lemma~\ref{lem:mm:analysis:repeated-geometric-tournament-mixing} and let
$t^{*}$ be the earliest time
$t^{*} \geq t + \mu + \max_{i \in [k]} \lambda_{i}$
such that $v$ starts a 2nd epoch tournament at time $t^{*}$ as part of a
$\Receiving$-phase.
Since each phase of $v$ consists of two tournaments and is a
$\Receiving$-phase with probability
$1 / 2$,
it follows by Lemmas
\ref{lem:mm:analysis:geometric-tournament-length-upper-bound} and
\ref{lem:mm:analysis:repeated-geometric-tournament-mixing} that
$t^{*} -t \leq O (\log^{3} n)$
w.h.p.;
condition hereafter on this event.

Let
$\hat{t} = t^{*} + \alpha \lceil \log n \rceil$
for a sufficiently large constant
$\alpha \in \Integers_{> 0}$
to be determined along the proof and let
$Y = d - \Degree^{\hat{t}}(v)$
be a random variable that captures the number of edges incident on $v$ that
become decided between time $t$ and time $\hat{t}$;
our goal is to prove that
$\Ex (Y) \geq \frac{c d}{\log n}$
for some (no matter how small) constant
$c > 0$.

Let $A$ be the event that
$v \in V^{t^{*}}$.
By definition,
$\neg A$
implies that
$\Degree^{\hat{t}} \leq \Degree^{t^{*}}(v) = 0$.
Therefore, if
$\Pr (\neg A) \geq \frac{c}{\log n}$
for any constant
$c > 0$,
then
\[
\Ex (Y)
\geq
\Ex (Y \mid \neg A) \cdot \Pr (\neg A)
\geq
d \cdot \frac{c}{\log n}
\, ,
\]
thus completing the proof.
Assume hereafter that
\begin{equation}
\label{eq:mm:analysis:good-degree-decreases:event-A}
\Pr (A)
>
1 - \frac{c}{\log n}
\end{equation}
for an arbitrarily small constant
$c > 0$.

For
$i \in [k]$,
let
$t_{i} = t^{*} - \lambda_{i}$
and let $A_{i}$ be the event that
$u_{i} \in V^{t_{i}}$.
Since
$Y \geq \sum_{i \in [k]} \Indicator_{\neg A_{i}}$,
it follows that
$\Ex (Y) \geq \Ex \left( \sum_{i \in [k]} \Indicator_{\neg A_{i}} \right)$,
thus if there exists some
$z > 0$
such that
$\Pr \left(
\sum_{i \in [k]} \Indicator_{\neg A_{i}} \geq \frac{z d}{\log n}
\right)
\geq
\frac{c}{z}$
for any constant
$c > 0$,
then
$\Ex (Y) \geq \frac{c d}{\log n}$
by Markov's inequality.
Assume hereafter that
\begin{equation}
\label{eq:mm:analysis:good-degree-decreases:events-neg-A-i}
\Pr \left(
\sum_{i \in [k]} \Indicator_{\neg A_{i}} \geq \frac{z d}{\log n}
\right)
<
\frac{c}{z}
\end{equation}
for any
$z > 0$
and for any arbitrarily small constant
$c > 0$.

We further define the following events for
$i \in [k]$:
Let $B_{i}$ be the event that node $u_{i}$ starts a phase at time $t_{i}$.
Let $C_{i}$ be the event that
(I)
a 1st epoch tournament $\T$ of $u_{i}$ ends at time $t^{*}$;
(II)
$\T$ belongs to a $\Proposing$-phase (of $u_{i}$);
and
(III)
$\T$ has a $Q_{i}$-winner.
Let $D_{i}$ be the event that port
$(u_{i}, \{ u_{i}, v \})$
is the $Q_{i}$-winner of a tournament of $u_{i}$ that ends at time $t^{*}$.

By Lemma~\ref{lem:mm:analysis:repeated-geometric-tournament-mixing},
we know that
$\Pr (B_{i})
\geq
\Omega \left( \frac{1}{\log d^{0}_{i}} \right)
\geq
\Omega \left( \frac{1}{\log n} \right)$,
whereas Lemma~\ref{lem:mm:analysis:geometric-tournament-exact-length-and-winner}
ensures that
$\Pr (C_{i} \mid B_{i})
\geq
\Omega (1)$.
Therefore,
\begin{equation}
\label{eq:mm:analysis:good-degree-decreases:event-B-i-and-C-i}
\Pr (B_{i} \land C_{i})
=
\Pr (C_{i} \mid B_{i}) \cdot \Pr (B_{i})
\geq
\frac{\phi}{\log n}
\end{equation}
for some universal constant
$\phi > 0$.
Moreover,
Lemma~\ref{lem:mm:analysis:geometric-tournament-exact-length-and-winner} also
ensures that
\begin{equation}
\label{eq:mm:analysis:good-degree-decreases:event-D-i-conditioned}
\Pr (D_{i} \mid B_{i} \land C_{i})
=
\frac{1}{|Q_{i}|}
=
\frac{1}{d_{i}}
\geq
\frac{1}{d}
\, ,
\end{equation}
where conditioned on
$B_{i} \land C_{i}$,
event $D_{i}$ is fully determined by the coin tosses of the ports incident on
$u_{i}$ during the time interval
$[t_{i}, t^{*})$.

Let $R_{i}$ be the event that
$\varProposal_{(u_{i} \{ u_{i}, v \})}^{t^{*}} = \Proposing_{0}$
and let
$R_{\lor} = \bigvee_{i \in [k]} R_{i}$.
The key observation now is that
\[
A_{i} \land B_{i} \land C_{i} \land D_{i}
\, \Longrightarrow \,
R_{i}
\qquad\text{and}\qquad
A \land R_{i}
\, \Longrightarrow \,
\varReceivingFlag_{(v, \{ u_{i}, v \})}^{t^{*} + 1} = \True
\, .
\]
Taking
$Q
=
\{ (u_{i}, \{ u_{i}, v \}) :
\varReceivingFlag_{(v, \{ u_{i}, v \})}^{t^{*} + 1} = \True \}$,
we can view the process of picking an incoming proposal that runs during the
2nd epoch of $v$ (the one that starts at time $t^{*}$) as two geometric
tournaments:
the first one is defined over the coin tosses of the ports in $\Ports(v)$;
the second one is defined over the coin tosses of the ports in $Q$ (that is, a
subset of the counter-ports of the ports in $\Ports(v)$).
As
$|\Ports(v)| \geq |Q|$,
we can apply Lemma~\ref{lem:mm:analysis:two-geometric-tournaments} to conclude
that there exists a universal constant
$c > 0$,
such that if
$A \land R_{\lor}$
occurs, then $v$ accepts a proposal during the phase with probability at
least $c$, thus
$\Degree^{\hat{t}}(v) = 0$
with probability at least $c$.
Therefore, our goal is to prove that
\[
\Pr \left(
A \land R_{\lor}
\right)
\geq
\frac{c}{\log n}
\]
for some (no matter how small) constant
$c > 0$.
As
\[
\Pr (A \land R_{\lor})
=
\Pr (A) - \Pr (A \land \neg R_{\lor})
\geq
\Pr (A) - \Pr (\neg R_{\lor})
\, ,
\]
we can apply \eqref{eq:mm:analysis:good-degree-decreases:event-A} to conclude
that to establish the assertion, it suffices to prove that
\begin{equation}
\label{eq:mm:analysis:good-degree-decreases:goal}
\Pr (R_{\lor})
\geq
\frac{c}{\log n}
\end{equation}
for an arbitrarily small constant
$c > 0$.

At this stage, the proof diverges to two cases according to the value of $d$:

\par\noindent
\textbf{Case 1:
$d \geq \frac{24}{\phi} \log n$.}
Inequality~\eqref{eq:mm:analysis:good-degree-decreases:event-B-i-and-C-i}
ensures that
$\Ex \left(
\sum_{i \in k} \Indicator_{B_{i} \land C_{i}}
\right)
\geq
\frac{\phi k}{\log n}
\geq
\frac{\phi d}{3 \log n}$.
Since the events in
$\{ B_{i} \land C_{i} \}_{i \in [k]}$
are mutually independent, it follows, by Chernoff's bound, that
\[
\Pr \left(
\sum_{i \in [k]} \Indicator_{B_{i} \land C_{i}} \leq \frac{\phi d}{6 \log n}
\right)
\leq
\exp \left( -\frac{\phi d}{24 \log n} \right)
\leq
e^{-1}
\, .
\]
By plugging
$z = \frac{\phi}{12}$
and
$c = \frac{\phi}{12} \left( \frac{1}{2} - e^{-1} \right)$
into \eqref{eq:mm:analysis:good-degree-decreases:events-neg-A-i}, we conclude
that
\[
\Pr \left(
\sum_{i \in [k]} \Indicator_{\neg A_{i}} \geq \frac{\phi d}{12 \log n}
\right)
<
\frac{1}{2} - e^{-1}
\, ,
\]
hence, be the union bound,
\[
\Pr \left(
\sum_{i \in [k]} \Indicator_{A_{i} \land B_{i} \land C_{i}}
>
\frac{\phi d}{12 \log n}
\right)
\geq
\Pr \left(
\sum_{i \in [k]} \Indicator_{B_{i} \land C_{i}} > \frac{\phi d}{6 \log n}
\land
\sum_{i \in [k]} \Indicator_{\neg A_{i}} < \frac{\phi d}{12 \log n}
\right)
>
\frac{1}{2}
\, .
\]
Conditioning on this event and taking
$J = \{ i \in [k] : A_{i} \land B_{i} \land C_{i} \}$,
inequality~\eqref{eq:mm:analysis:good-degree-decreases:event-D-i-conditioned}
implies that the probability of
$\bigwedge_{j \in J} \neg D_{i}$
is bounded from above by
\[
\left( 1 - \frac{1}{d} \right)^{|J|}
<
\left( 1 - \frac{1}{d} \right)^{\frac{\phi d}{12 \log n}}
<
\exp \left( -\frac{\phi}{12 \log n} \right)
<
1 - \frac{\phi}{24 \log n}
\, ,
\]
where the last transition holds as
$e^{-z} < 1 - \frac{z}{2}$
for all
$0 < z < 1$.
We can now establish \eqref{eq:mm:analysis:good-degree-decreases:goal} by
developing
\begin{align*}
\Pr (R_{\lor})
\geq \, &
\Pr \left(
\bigvee_{i \in [k]} A_{i} \land B_{i} \land C_{i} \land D_{i}
\right)
\\
\geq \, &
\Pr \left(
\bigvee_{i \in [k]} A_{i} \land B_{i} \land C_{i} \land D_{i}
\mid
\sum_{i \in [k]} \Indicator_{A_{i} \land B_{i} \land C_{i}}
>
\frac{\phi d}{12 \log n}
\right)
\cdot
\\
&
\Pr \left(
\sum_{i \in [k]} \Indicator_{A_{i} \land B_{i} \land C_{i}}
>
\frac{\phi d}{12 \log n}
\right)
\\
> \, &
\frac{\phi}{24 \log n} \cdot \frac{1}{2}
=
\frac{\phi}{48 \log n}
\, .
\end{align*}

\par\noindent
\textbf{Case 2:
$d < \frac{24}{\phi} \log n$.}
Since the events in
$\{ B_{i} \land C_{i} \}_{i \in [k]}$
are mutually independent, it follows, by
\eqref{eq:mm:analysis:good-degree-decreases:event-B-i-and-C-i}, that
\[
\Pr \left(
\bigwedge_{i \in [k]} \neg (B_{i} \land C_{i})
\right)
\leq
\left( 1 - \frac{\phi}{\log n} \right)^{k}
\leq
\left( 1 - \frac{\phi}{\log n} \right)^{d / 3}
<
\exp \left( -\frac{\phi d}{3 \log n} \right)
<
1 - \frac{\phi d}{48 \log n}
\, ,
\]
where the last transition holds as
$e^{-z} < 1 - \frac{z}{2}$
for all
$0 < z < 1$.
By plugging
$z = \frac{\log n}{d}$
and
$c = \frac{\phi}{96}$
into \eqref{eq:mm:analysis:good-degree-decreases:events-neg-A-i}, we conclude
that
\[
\Pr \left(
\bigvee_{i \in [k]} \neg A_{i}
\right)
=
\Pr \left( \sum_{i \in [k]} \Indicator_{\neg A_{i}} \geq 1 \right)
<
\frac{\phi d}{96 \log n}
\, .
\]
Therefore,
\begin{align*}
\Pr \left(
\bigvee_{i \in [k]} A_{i} \land B_{i} \land C_{i}
\right)
\geq \, &
\Pr \left(
\bigwedge_{i \in [k]} A_{i}
\land
\bigvee_{i \in [k]} B_{i} \land C_{i}
\right)
\\
= \, &
\Pr \left(
\bigwedge_{i \in [k]} A_{i}
\right)
-
\Pr \left(
\bigwedge_{i \in [k]} A_{i} \land \neg \bigvee_{i \in [k]} B_{i} \land C_{i}
\right)
\\
\geq \, &
\Pr \left(
\bigwedge_{i \in [k]} A_{i}
\right)
-
\Pr \left(
\neg \bigvee_{i \in [k]} B_{i} \land C_{i}
\right)
\\
> \, &
\left( 1 - \frac{\phi d}{96 \log n} \right)
-
\left( 1 - \frac{\phi d}{48 \log n} \right)
=
\frac{\phi d}{96 \log n}
\, .
\end{align*}
We can now establish \eqref{eq:mm:analysis:good-degree-decreases:goal} by
developing
\begin{align*}
\Pr (R_{\lor})
\geq \, &
\Pr \left(
\bigvee_{i \in [k]} A_{i} \land B_{i} \land C_{i} \land D_{i}
\right)
\\
\geq \, &
\Pr \left(
\bigvee_{i \in [k]} A_{i} \land B_{i} \land C_{i} \land D_{i}
\mid
\bigvee_{i \in [k]} A_{i} \land B_{i} \land C_{i}
\right)
\cdot
\Pr \left(
\bigvee_{i \in [k]} A_{i} \land B_{i} \land C_{i}
\right)
\\
> \, &
\frac{1}{d} \cdot \frac{\phi d}{96 \log n}
=
\frac{\phi}{96 \log n}
\, ,
\end{align*}
where the penultimate transition follows from
\eqref{eq:mm:analysis:good-degree-decreases:event-D-i-conditioned}.
\end{proof}

%%%%%%%%%%%%%%%%%%%%%%%%%%%%%%%%%%%%%%%
\subsubsection{Stochastic Analysis of the Tournaments}
\label{sec:mm:analysis:tournaments}
%%%%%%%%%%%%%%%%%%%%%%%%%%%%%%%%%%%%%%%
%YE: streamline (Oren)
Henceforth, for integer
$\kappa > 0$
and
$Q \subseteq \{ X_{1}, \dots, X_{\kappa} \}$,
we let $L_\kappa(Q)$ and $J_\kappa(Q)$ be the length and winner of a
tournament $\T$ restricted only to those random variables in $Q$.
That is, $L_\kappa(Q) = \max_{X_i \in Q} X_i$ and $J_\kappa(Q)$ is the random
variable $X_j \in Q$ such that $X_j > X_i$ for all $X_i \in Q - \{ X_j \}$ if
such exists; if not, we set $J_\kappa(Q) = \bot$.
We also abbreviate $L_\kappa \equiv L_\kappa(\{ X_{1}, \dots, X_{\kappa} \})$ and
$J_\kappa \equiv J_\kappa (\{ X_{1}, \dots, X_{\kappa} \})$, and write
$Q^\rmc := \{ X_{1}, \dots, X_{\kappa} \} - Q$.
Finally, we let
\begin{equation*}
	m_\kappa := \log \kappa \,,	
\end{equation*}
for the typical value of $L_\kappa$ up to $O(1)$. 

The following lemma upper bounds the right and left tails of $L_\kappa$ around $m_\kappa$ and, in particular, establishes exponential tightness.
\begin{lemma}
\label{l:4.8}
For all integer $\kappa \geq 1$ and $t \geq 0$,
\begin{equation}
\label{eq:mm:analysis:oren-1}
	\Pr(L_\kappa > m_{\kappa} + t) \leq 4 \cdot 2^{-t}
\end{equation}
and
\begin{equation}
\label{eq:mm:analysis:oren-2}
	\Pr(L_\kappa < m_{\kappa} - t) \leq \rme^{-2^{t}} \,.
\end{equation}
In particular, the sequence of random variables $\{L_\kappa - m_{\kappa}\}$ is exponentially tight.
\end{lemma}
\begin{proof}
By the union bound the probability in \eqref{eq:mm:analysis:oren-1} is upper bounded by
\begin{equation*}
	\kappa \Pr (X_1 > m_\kappa + t) \leq \kappa 2^{-(m_\kappa + t - 2)} \,, 
\end{equation*}
which is bounded by the right hand side in \eqref{eq:mm:analysis:oren-1}.

On the other hand, by independence, the probability in \eqref{eq:mm:analysis:oren-2} is equal to
\begin{equation*}
\big(\Pr (X_1 < m_\kappa - t)\big)^\kappa
\leq \Big( 1 - 2^{-(m_\kappa - t)} \Big)^\kappa
\leq \rme^{-\kappa 2^{-(m_\kappa - t)}} \,,
\end{equation*}
which is upper bounded by the right hand side of \eqref{eq:mm:analysis:oren-2}.
\end{proof}

The next lemma lower bounds the right tail of $L_\kappa$ around $m_\kappa$, together with the value of $J_\kappa(Q)$ for all $Q$.
\begin{lemma}
\label{l:4.9}
There exists $c>0$, such that for all $\kappa \geq 1$ integer, all $t \geq 0$ such that $m_\kappa + t \geq 3$ is an integer, all $\emptyset \subsetneq Q \subseteq \{ X_{1}, \dots, X_{\kappa} \}$ and all $X_j \in Q$,
\begin{equation}
\label{eq:mm:analysis:oren-3}
\Pr \Big( L_\kappa = m_\kappa + t \,,\,\, J_\kappa(Q) = X_j \Big) \geq c\frac{1}{|Q|} 2^{-t} \, .
\end{equation}
In particular, under the same conditions,
\begin{equation*}
	\Pr \Big( L_\kappa = m_\kappa + t \,,\,\, J_\kappa(Q) \neq \bot \Big) \geq c 2^{-t} \,.
\end{equation*}
\end{lemma}
\begin{proof}
The second statement follows from the first by summation over $X_j \in Q$. Turning to the first, assume initially that $Q = \{ X_{1}, \dots, X_{\kappa} \}$. Then the desired probability is equal to
\begin{equation*}
\Pr \big(X_j = m_\kappa + t\big) \prod_{i \neq j} \Pr \big(X_i < m_\kappa + t \big)
= 2^{-(m_\kappa+t-2)} \Big(1-2^{-(m_\kappa+t-2)}\Big)^{\kappa-1}
\geq \frac{1}{\kappa} \times 2^{-t} \rme^{-2^{-(t-3)}} \,,
\end{equation*}	
whenever $2^{-t+2}/\kappa$ is small enough, so that we can use the inequality
$1-x > \rme^{-2x}$. The right hand side is larger than the right hand side in
\eqref{eq:mm:analysis:oren-3} with a suitable $c > 0$. If $2^{-t+2}/\kappa$ is not small enough, it must be that $t < t_0$ and $\kappa < \kappa_0$ for some fixed $t_0$, $\kappa_0$. In all these cases, it can be easily checked by hand that the desired probability is lower bounded by a positive constant, thanks to the requirement on $m_{\kappa}+t$. Modifying $c$ as needed, this shows \eqref{eq:mm:analysis:oren-3}.

Now for general $Q$, if $|Q| \geq \kappa/2$, use independence of the $X$-s to
lower bound the probability in \eqref{eq:mm:analysis:oren-3} by
\begin{equation*}
\Pr \Big( L_\kappa(Q) = m_{|Q|} + (t + m_\kappa - m_{|Q|}) \,,\,\, J_\kappa(Q) = X_j \Big)
\Pr \Big( L_\kappa(Q^\rmc) = \lceil m_{|Q^\rmc|} \rceil \Big) \,,
\end{equation*}
since $\lceil m_{|Q^\rmc|} \rceil \leq m_\kappa + t$. Since $m_{\kappa} - m_{|Q|} \in [0,1]$ and 
the $X$-s are identically distributed, the first probability is at least $(c/2) |Q|^{-1} 2^{-t}$ and the second at least $c/2$, by what we have shown before. On the other hand, if $|Q| < \kappa/2$, we lower bound the desired probability by
\begin{equation*}
\Pr \Big( L_\kappa(Q^\rmc) = m_{|Q^{\rmc}|} + (t + m_\kappa - m_{|Q^\rmc|}) \Big)
\Pr \Big( L_\kappa(Q) = \lceil m_{|Q|} \rceil ,\, J_\kappa(Q) = X_j \Big) \,.
\end{equation*}
Again, the first probability is at least $(c/2) 2^{-t}$ and the second $(c/2)
|Q|^{-1}$. In both cases the product is lower bounded by the right hand side of \eqref{eq:mm:analysis:oren-3} with $c^2/4 > 0$ in place of $c$.
\end{proof}

We are now ready to prove the first three of the four lemmas.
\begin{proof}[Proof of Lemma~\ref{lem:mm:analysis:geometric-tournament-length-upper-bound}]
Plug in $t=\log \kappa$ in the first statement of Lemma~\ref{l:4.8} and note that $\kappa \leq n$.
\end{proof}

\begin{proof}[Proof of Lemma~\ref{lem:mm:analysis:geometric-tournament-exact-length-and-winner}]
Take $\lambda := \lceil m_\kappa \rceil \vee 3$ and use Lemma~\ref{l:4.9} with $t := \lambda - m_\kappa \in [0, 3]$.
\end{proof}

\begin{proof}[Proof of Lemma~\ref{lem:mm:analysis:two-geometric-tournaments}]
Take $\lambda := \lceil m_\kappa \rceil \vee 3$, $\lambda' := (\lceil m_{\kappa'} \rceil \vee 3) + \alpha$,
so that $t := \lambda - m_\kappa \in [0,3]$ and $t' := \lambda' - m_{\kappa'} \in [0, 3+\alpha]$.
Then the desired probability is lower bounded by 
\begin{equation*}
	\Pr \big(L_\kappa = \lambda\,,\,\, J_\kappa(Q) \neq \bot \big)
	\Pr \big(L_{\kappa'} = \lambda' \big) \,.
\end{equation*}
By Lemma~\ref{l:4.9}, the first term is lower bounded by a universal positive constant and the second by an $\alpha$ dependent positive one.
\end{proof}

Next we turn to the last lemma. To this end, we shall need a local-central-limit type result. 
\begin{lemma}[Local Central Limit Theorem Approximation] \label{lem:lclt}
There exists $C < \infty$ such that for all integers $\kappa \geq 1$, $h \geq 1$ and $s$, if $S$ is a repeated geometric tournament with $\kappa$ random variables, then
\begin{equation*}
\left| \Pr(S_h = s) - \frac{1}{\sigma \sqrt{2\pi h}} \rme^{-\frac{(s - h \mu)^2}{2h\sigma^2}} \right| 
\leq \frac{C}{h}
\end{equation*}
where $\mu$ and $\sigma^2$ are the mean and variance of $L_\kappa$. Moreover, for all $\kappa \geq 1$,
\begin{equation*}
	|\mu - m_\kappa| \leq C \quad ; \qquad \frac{1}{C} \leq \sigma \leq C
\end{equation*}

\end{lemma}
\begin{proof}
The desired statement does not change if we replace $S, \sigma, \mu$ by $S', \sigma', \mu'$, defined as the former only with $L_\kappa' := L_\kappa - \lfloor m_\kappa \rfloor$ in place of $L_\kappa$. Then, by Theorem 1.2 from~\cite{siripraparat2021improvement}
%YE: one of the reviewers asked to cite the theorem (I guess that they meant to present its statement). 
the desired difference is bounded in absolute value
for all $h,s$, by a constant times
\begin{equation}
\label{eq:mm:analysis:oren-4}
\frac{\exp(-c\, \tau^2 \alpha)}{\tau \alpha} + \frac{1}{h} \frac{\sigma_3^3}{\sigma_2^4} \,,
\end{equation}
where
\begin{equation}
\label{eq:mm:analysis:oren-5}
\alpha := h \cdot\sum_{l = -\infty}^{\infty} \Pr \big(L'_\kappa = l\big)\, \Pr \big(L'_\kappa = l+1\big) \quad ; \qquad  \tau := \sigma_3^{-1} h^{-1/3}
\quad; \qquad \sigma_p := \|L'_\kappa\|_p  \,,
\end{equation}
and $c > 0$.

By the exponential tightness of $L'_\kappa$, as shown by Lemma~\ref{l:4.8} and the second part of Lemma~\ref{l:4.9}, all positive moments of $L'_\kappa$ and its variance are bounded from above and away from zero uniformly in $\kappa$. 
These two lemmas also show that the sum in \eqref{eq:mm:analysis:oren-5} is
bounded away from zero uniformly in $\kappa$. This bounds the expression in \eqref{eq:mm:analysis:oren-4} by $C h^{-1}$ for a properly chosen $C < \infty$, which is uniform in $\kappa$, and also proves the desired bound on $\sigma$. It also gives the bound on $\mu$ via an application of Jensen's Inequality.
\end{proof}

\begin{proof}[Proof of Lemma~\ref{lem:mm:analysis:repeated-geometric-tournament-mixing}]
Assume without loss of generality that $j=s=0$.
We can also assume that $z$ is arbitrarily large.
Indeed, if it is not, then $\kappa$ must also be small and then one can check
by hand that the desired probability is uniformly positive for all such
$\kappa$-s and $z$-s.
Set $h_0 = \lceil z/\mu \rceil$, $h_1 = \lceil h_0 + \sigma \sqrt{h_0}/\mu \rceil$ where $\mu$, $\sigma$ are as in Lemma~\ref{lem:lclt}. By Lemma~\ref{lem:lclt},
\begin{equation*}
	\Pr \big(\exists h \in 2\Integers_{\geq 0} :\: S_h = z \big) 
\geq \sum_{\substack{h=h_0\\h \in 2\Integers}}^{h_1} \Pr \big(S_h = z \big) 
\geq c \,\frac{h_1 - h_0}{\sqrt{h_0}} \geq c' \frac{\sigma}{\mu} \,.
\end{equation*}
Above, we have used that $(z - h \mu)^2/(h\sigma^2)$ is uniformly upper bounded, which holds whenever $z$ and hence $h_0$ are sufficiently large.
The result follows, in view of the estimates on $\mu$ and $\sigma$ in Lemma~\ref{lem:lclt}. The case of odd $h$ is proved exactly in the same way.
\end{proof}

%%%%%%%%%%%%%%%%%%%%%%%%%%%%%%%%%%%%%%%%%%%%%%%%%%%%%%%%%%%%%%%%%%%%%%%%%%%%%%
\section{Sinkless Orientation}
\label{sec:so}
%%%%%%%%%%%%%%%%%%%%%%%%%%%%%%%%%%%%%%%%%%%%%%%%%%%%%%%%%%%%%%%%%%%%%%%%%%%%%%
An \emph{orientation} of the graph
$G = (V, E)$
is a set
$\vec{E} \subseteq V \times V$
of size
$|\vec{E}| = |E|$
such that
$(u, v) \in \vec{E} \Longrightarrow \{ u, v \} \in E$;
an edge
$\{ u, v \} \in E$
is regarded as oriented outward of $u$ and toward $v$ in $\vec{E}$ if
$(u, v) \in \vec{E}$.
The goal in the \emph{sinkless orientation (SO)} problem is to construct an
orientation $\vec{E}$ of $G$ such that for every node
$v \in V$
of degree
$\Degree_{G}(v) \geq 3$,
the out-degree of $v$ in the digraph
$\vec{G} = (V, \vec{E})$
is at least $1$ (i.e., $v$ is not a sink in
$\vec{G}$).
Under the UP model, this is translated to using
$\{ \TO, \FROM \}$
as the output label set so that an edge
$\{ u, v \} \in E$
is oriented toward $v$ in a configuration $C$ if
$\omega(C((u, \{ u, v \}))) = \FROM$
and
$\omega(C((v, \{ u, v \}))) = \TO$.

\begin{theorem}
\label{thm:so}
There exists a self-stabilizing UP algorithm that solves the SO problem and
stabilizes in
$O(\log^{2} n)$
time w.h.p.
\end{theorem}

The algorithm promised in Theorem~\ref{thm:so} is presented in
Section~\ref{sec:so:algorithm}, first at a high level and then, from the
perspective of the individual ports.
Section~\ref{sec:so:analysis} is then dedicated to proving the correctness of
our algorithm and bounding its stabilization time, thus establishing
Theorem~\ref{thm:so}.

%%%%%%%%%%%%%%%%%%%%%%%%%%%%%%%%%%%%%%%
\subsection{The SO Algorithm}
\label{sec:so:algorithm}
%%%%%%%%%%%%%%%%%%%%%%%%%%%%%%%%%%%%%%%

%%%%%%%%%%%%%%%%%%%%%%%%%%%%%%%%%%%%%%%
\subsubsection*{High-Level Description}
%%%%%%%%%%%%%%%%%%%%%%%%%%%%%%%%%%%%%%%
As a preliminary step, our SO algorithm orients all edges arbitrarily and
selects three incident edges
$e_{1}, e_{2}, e_{3}$
for each node
$v \in V$
of degree
$\Degree(v) \geq 3$,
referred to in the scope of this section as the
\emph{designated edges} of $v$.
Following that, the goal of the algorithm is to ensure that at least one of
the designated edges of $v$ is oriented outward of $v$.
To this end, whenever $v$ becomes a sink with respect to its designated edges,
node $v$ flips the orientation of edge
$e_{i}$ for some
$i \in \{ 1, 2, 3 \}$.
The choice of edge $e_{i}$, referred to as the \emph{next flip} of $v$, is
made in the previous round (as a preparation for the event that $v$ becomes a
sink) according to the following rule:
if exactly $2$ of $v$'s designated edges
$e_{1}, e_{2}, e_{3}$
are oriented toward $v$, then $v$ picks the next flip uniformly at
random among these two designated edges;
otherwise (zero, one, or three of $v$'s designated edges are oriented toward
$v$), $v$ picks the next flip uniformly at random among its three designated
edges.

%%%%%%%%%%%%%%%%%%%%%%%%%%%%%%%%%%%%%%%
\subsubsection*{Port-Level Description}
%%%%%%%%%%%%%%%%%%%%%%%%%%%%%%%%%%%%%%%
The edge orientation is maintained by means of variables
$\varLocalOrientation_{p} \in \{ 0, 1 \}$,
referred to as the \emph{local orientation} variables, that each port
$p \in \Ports$
maintains.
An edge
$e = \{ v_{1}, v_{2} \} \in E$
is considered to be oriented toward $v_{i}$,
$i \in \{ 1, 2 \}$,
if
$\varLocalOrientation_{(v_{i}, \{ v_{1}, v_{2} \})} = 1$
and
$\varLocalOrientation_{(v_{3 - i}, \{ v_{1}, v_{2} \})} = 0$.
As long as edge $e$ is unoriented (that is,
$\varLocalOrientation_{(v_{1}, \{ v_{1}, v_{2} \})} =
\varLocalOrientation_{(v_{2}, \{ v_{1}, v_{2} \})}$),
port
$(v_{i}, \{ v_{1}, v_{2} \})$,
$i \in \{ 1, 2 \}$,
updates its local orientation variable by picking
$\varLocalOrientation_{(v_{i}, \{ v_{1}, v_{2} \})} \sim \UnifDist(\{ 0, 1 \})$.
The algorithm is designed so that once $e$ becomes oriented, it remains
oriented, although the orientation may flip if certain conditions,
presented in the sequel, are satisfied.

Consider a node
$v \in V$
and an incident edge
$e = \{ u, v \}$
and let
$p = (v, e)$
be the port incident on $v$ and $e$.
Port $p$ maintains a variable
$\varDesignated_{p} \in \{ 1, 2, 3, \bot \}$,
referred to as the \emph{designated edge} of $p$.
The semantics of this variable is as follows:
$\varDesignated_{p} = i$
for some
$i \in \{ 1, 2, 3 \}$
indicates that $e$ is the $i$-th designated edge of $v$;
$\varDesignated_{p} = \bot$
indicates that $e$ is not one of the three designated edges of $v$.
The mechanism that controls the selection of exactly
$\min \{ 3, \Degree(v) \}$
designated edges for $v$ works as follows:
\begin{itemize}

\item
If
$\varDesignated_{p} \in \{ 1, 2 ,3 \}$
and there exists
$p' \in \Siblings(p)$
such that
$\varDesignated_{p'} = \varDesignated_{p}$,
then $p$ resets
$\varDesignated_{p} \gets \bot$
with probability
$1 / 2$;
and
keeps the current value of $\varDesignated_{p}$ (formally
$\varDesignated_{p} \gets \varDesignated_{p}$)
with probability
$1 / 2$.

\item
Otherwise, if
$\varDesignated_{p} = \bot$
and
$\{ 1, 2, 3\} \nsubseteq \{ \varDesignated_{p'} : p' \in \Siblings(p) \}$,
then $p$ assigns
$\varDesignated_{p} \gets i_{\min}$,
where $i_{\min}$ is the smallest
$i \in \{ 1, 2, 3 \}$
such that
$i \notin \{ \varDesignated_{p'} : p' \in \Siblings(p) \}$.

\item
Otherwise, $p$ keeps the current value of $\varDesignated_{p}$ (formally
$\varDesignated_{p} \gets \varDesignated_{p}$).

\end{itemize}

The aforementioned mechanism guarantees that for each node
$v \in V$,
once the designated edge variables of the ports in $\Ports(v)$ have
stabilized, there exists exactly one port
$p \in \Ports(v)$
with
$\varDesignated_{p} = i$
for each
$1 \leq i \leq \min \{ 3, \Degree(v) \}$
--- we subsequently refer to this port $p$ as \emph{the $i$-th designated
port} of $v$ and denote it by
$p_{i}(v)$.
The ports in
$p \in \Ports(v)$
can distinguish between the cases
$\Degree(v) = 1$,
$\Degree(v) = 2$,
and
$\Degree(v) \geq 3$
based on the maximum index $i$ for which $p_{i}(v)$ exists (notice that
$p_{1}(v)$ must exist as
$\Degree(v) \geq 1$).
Assuming that
$\Degree(v) \geq 3$,
let
$D(v) = \{ i \in \{ 1, 2, 3 \} : \varLocalOrientation_{p_{i}(v)} = 1 \}$;
node $v$ is regarded as a \emph{designated sink} if $|D(v)| = 3$.

Consider a node
$v \in V$
of degree
$\Degree(v) \geq 3$.
Port $p_{1}(v)$ maintains a variable
$\varNextFlip_{p_{1}(v)} \in \{ 1, 2, 3 \}$,
referred to as the \emph{next flip} of $v$, that determines the designated
edge whose orientation is to be flipped in case $v$ becomes a designated sink.
Variable
$\varNextFlip_{p_{1}(v)}$
is updated in each round according to the
following simple rule:
if
$|D(v)| = 2$,
then
$\varNextFlip_{p_{1}(v)} \sim \UnifDist(D(v))$;
otherwise
($|D(v)| = 0$,
$|D(v)| = 1$,
or
$|D(v)| = 3$),
$\varNextFlip_{p_{1}(v)} \sim \UnifDist(\{ 1, 2, 3 \})$.

The edge orientation flipping mechanism is controlled by a Boolean
variable
$\varFlipFlag_{p_{i}(v)} \in \{ \True, \False \}$,
referred to as the \emph{flip flag}, that each designated port $p_{i}(v)$,
$i \in \{ 1, 2, 3 \}$,
maintains.
Port $p_{i}(v)$ turns its flip flag on, setting
$\varFlipFlag_{p_{i}(v)} \gets \True$
if and only if
(I)
$D(v) = 3$,
that is, $v$ is a designated sink;
(II)
$\varNextFlip_{p_{1}(v)} = i$;
and
(III)
$\varFlipFlag_{p_{i}(v)} = \False$.

An orientation flip of an oriented edge
$e \in E$
is triggered by, and only by, a turned on flip flag:
Consider a port
$p \in \Ports$
and assume that the edge incident on $p$ is oriented, i.e.,
$\varLocalOrientation_{p} \neq \varLocalOrientation_{\CounterPort{p}}$.
Port $p$ changes the value of its local orientation variable from
$\varLocalOrientation_{p} = 1$
to
$\varLocalOrientation_{p} = 0$
(resp., from
$\varLocalOrientation_{p} = 0$
to
$\varLocalOrientation_{p} = 1$)
if and only if
$\varFlipFlag_{p} = \True$
(resp.,
$\varFlipFlag_{\CounterPort{p}} = \True$).

To complete the algorithm's description, the output label associated with each
port
$p \in \Ports$
is determined directly from its local orientation variable:
the output label associated with $p$ is $\FROM$ or $\TO$ if
$\varLocalOrientation_{p} = 0$
or
$\varLocalOrientation_{p} = 1$,
respectively.

%%%%%%%%%%%%%%%%%%%%%%%%%%%%%%%%%%%%%%%
\subsection{Analysis}
\label{sec:so:analysis}
%%%%%%%%%%%%%%%%%%%%%%%%%%%%%%%%%%%%%%%
Throughout the analysis, we fix an arbitrary initial configuration $C^{0}$ and
consider the (random) execution
$\eta = \{ C^{t} \}_{t \geq 0}$
of the SO algorithm starting from $C^{0}$.
A time
$t \in \Integers_{> 0}$
is said to be \emph{clean} with respect to $\eta$ if the following conditions
hold:
\begin{enumerate}

\item
\label{so:analysis:clean-condition:oriented-edges}
All edges are oriented in $C^{t}$,
i.e.,
$\varLocalOrientation_{(v_{1}, e)}^{t}
\neq
\varLocalOrientation_{(v_{2}, e)}^{t}$
for every
$e = \{ v_{1}, v_{2} \} \in E$.

\item
\label{so:analysis:clean-condition:designated-ports}
For every node
$v \in V$
and index
$1 \leq i \leq \min \{ 3, \Degree(v) \}$,
there exists exactly one port
$p \in \Ports(v)$
such that
$\varDesignated_{p}^{t} = i$.

\item
\label{so:analysis:clean-condition:flip-flag}
For every node
$v \in V$
and port
$p \in \Ports(v)$
such that
$\varDesignated_{p}^{t} \in \{ 1, 2, 3 \}$,
if
$\varFlipFlag_{p}^{t} = \True$,
then
(I)
$|\{
p' \in \Ports(v)
:
\varDesignated_{p'}^{t} \in \{ 1, 2, 3 \}
\land
\varLocalOrientation_{p'}^{t} = 1
\}| = 3$;
and
(II)
$\varFlipFlag_{p'}^{t} = \False$
for every
$p' \in \Siblings(p)$
such that
$\varDesignated_{p'}^{t} \in \{ 1, 2, 3 \}$.

\end{enumerate}
The following lemma allows us to designate a clean suffix of execution $\eta$;
the subsequent analysis is then dedicated to analyzing that suffix.

\begin{lemma}
\label{lem:so:analysis:clean-time}
There exists a  time
$\tClean \in \Integers_{> 0}$
such that all times
$t \geq \tClean$
are clean with respect to $\eta$ a.s.
Moreover,
$\tClean \leq O (\log^{2} n)$
w.h.p.
\end{lemma}
\begin{proof}
Consider an edge
$e \in E$.
The update rule of the local orientation variables ensures that once $e$
becomes oriented, it remains oriented at all subsequent times.
Moreover, if $e$ is not oriented at time
$t > 0$,
then $e$ becomes oriented in round $t$ with probability
$1 / 2$,
independently.
Therefore, there exists a time
$t^{1} \in \Integers_{> 0}$
such that condition~(\ref{so:analysis:clean-condition:oriented-edges}) holds
at all times
$t \geq t^{1}$
and
$t^{1} \leq O (\log n)$
w.h.p.;
condition hereafter on this event.

Consider a node
$v \in V$
and an index
$1 \leq i \leq \min \{ 3, \Degree(v) \}$
and suppose that the designated port
$p_{i'}(v)$
has already been elected for all
$1 \leq i' < i$.
The process behind the election of the designated port
$p_{i}(v)$
is equivalent to an invocation of the LLE mechanism (see
Section~\ref{sec:prelim}) and can be analyzed in the same way to
conclude that there exists a time
$t^{2} \in \Integers_{> 0}$
such that condition~(\ref{so:analysis:clean-condition:designated-ports}) holds
at all times
$t \geq t^{2}$
and
$t^{2} \leq O (\log^{2} n)$
w.h.p.;
condition hereafter on this event.

The proof is completed by observing that
condition~(\ref{so:analysis:clean-condition:flip-flag}) holds
(deterministically) at all times
$t > \max \{ t^{1}, t^{2} \}$.
\end{proof}

Let
$\tClean \in \Integers_{> 0}$
be the time promised in Lemma~\ref{lem:so:analysis:clean-time}.
Our goal in the remainder of this section is to prove that starting from time
$\tClean$, the algorithm stabilizes within
$O (\log^{2} n)$
rounds w.h.p.
To this end, we introduce the following additional definitions.

\begin{definition*}[%
$p_{i}(v)$,
$e_{i}(v)$,
two-sided designated edge]
Consider a node
$v \in V$
and an index
$1 \leq i \leq \min \{ 3, \Degree(v) \}$
and recall that port
$p = (v, \{ u, v \}) \in \Ports(v)$
(resp., edge
$\{ u, v \}$)
is regarded as the $i$-th designated port (resp., edge) of $v$ at time $t$ if
$\varDesignated_{p}^{t} = i$.
Let $p_{i}(v)$ (resp., $e_{i}(v)$) be the $i$-th designated port (resp., edge)
of $v$ at time $\tClean$;
observe that $p_{i}(v)$ (resp., $e_{i}(v)$) remains the $i$-th designated port
(resp., edge) of $v$ at all times
$t \geq \tClean$
as the values of the designated edge variables do not change from time
$\tClean$ onward.
We say that an edge
$e = (v_{1}, v_{2}) \in E$
is \emph{two-sided designated} if there exist some
$1 \leq i_{1} \leq \min \{ 3, \Degree(v_{1}) \}$
and
$1 \leq i_{2} \leq \min \{ 3, \Degree(v_{2}) \}$
such that
$e = e_{i_{1}}(v_{1})$
and
$e = e_{i_{2}}(v_{2})$.
\end{definition*}

\begin{definition*}[%
$I^{t}(v)$,
designated sink,
susceptible node]
For a node
$v \in V$
of degree
$\Degree(v) \geq 3$
and a time
$t > \tClean$,
let
$I^{t}(v)
=
\{
e_{i}(v) : i \in \{ 1, 2, 3 \} \land \varLocalOrientation_{p_{i}(v)}^{t} = 1
\}$
be the set of designated edges of $v$ oriented toward $v$ at time $t$.
We say that $v$ is a \emph{designated sink} at time $t$ if
$|I^{t}(v)| = 3$.
We say that $v$ is \emph{susceptible} at time $t$ if
$|I^{t}(v)| = 2$.
\end{definition*}

The following observation states that after time $\tClean$, no node serves as
a designated sink for more than two consecutive rounds.

\begin{observation}
\label{obs:so:analysis:designated-sink-interval}
Consider a node
$v \in V$
of degree
$\Degree(v) \geq 3$
and a time
$t > \tClean$.
If $v$ is a designated sink throughout the time interval
$[t, t')$,
then
$t' \leq t + 2$.
\end{observation}

It is convenient to view the execution after time $\tClean$ through the lens
of a virtual \emph{token passing process} over the graph $G$, defined by
placing a token on node
$v \in V$
at time
$t > \tClean$
if and only if $v$ is a designated sink at time $t$ (recall that by
definition, this requires that
$\Degree(v) \geq 3$).
We shall bound the time it takes for the algorithm to stabilize by bounding
the time it takes for all tokens to be deleted from $G$.

Consider a node
$v \in V$
that holds a token $\tau$ at time
$\hat{t} > \tClean$
and notice that Observation~\ref{obs:so:analysis:designated-sink-interval}
guarantees that $v$ cannot hold $\tau$ for more than two consecutive rounds.
In particular, node $v$ disposes of $\tau$ in round
$\hat{t} \leq t \leq \hat{t} + 1$
by flipping the orientation of some edge
$e = \{ u, v \} \in I^{t}(v)$
(from the perspective of the ports, this means that
$\varLocalOrientation_{(v, e)}^{t}
=
1
=
\varLocalOrientation_{(u, e)}^{t + 1}$
and
$\varLocalOrientation_{(u, e)}^{t}
=
0
=
\varLocalOrientation_{(v, e)}^{t + 1}$,
thus turning into a susceptible node.
From the perspective of $\tau$, edge $e$'s orientation flip in round $t$
results in one of three possible outcomes:
\begin{itemize}

\item
A \emph{token shifting} outcome:
We think of token $\tau$ as being shifted from $v$ to $u$ in round $t$ (so that
$\tau$ is placed on $u$ at time
$t + 1$) if
(I)
$\Degree(u) \geq 3$;
(II)
$e$ is two-sided designated;
and
(III)
$u$ is susceptible at time $t$ (which means that $e$ is the only designated
edge of $u$ that is oriented outward of $u$ at time $t$).

\item
A \emph{token deletion} outcome:
We think of token $\tau$ as being deleted from the graph in round $t$ if at
least one of the following conditions is satisfied:
(I)
$\Degree(u) < 3$;
(II)
$e$ is not two-sided designated;
or
(III)
$u$ is not a designated sink at time
$t + 1$.\footnote{%
Notice that a token deletion outcome does not rule out the possibility that
node $u$ becomes a designated sink at some time
$t' > t + 1$;
this however is reflected in the token passing process by $u$ holding
another token
$\tau' \neq \tau$,
which, in the grand scheme of things, allows us to argue that the total number
of tokens decreases by a factor of (at least) $2$.}

\item
A \emph{token merging} outcome:
We think of token $\tau$ as being merged with other (at
least one) tokens in round $t$, if
(I)
$\Degree(u) \geq 3$;
(II)
$e$ is a two-sided designated;
(III)
$|I^{t}(u)| < 2$
(that is, $u$ is neither susceptible, nor a designated sink, at time $t$);
and
(IV)
$u$ is a designated sink at time
$t + 1$.

\end{itemize}

We establish Theorem~\ref{thm:so} by proving the following lemma.

\begin{lemma}
\label{lem:so:analysis:number-of-tokens-decreases}
Fix some time
$t > \tClean$
and the configuration $C^{t}$ and let $\kappa$ be the number of tokens in the
graph at time $t$.
Then, at time
$t + O (\log n)$,
the graph includes at most
$\kappa / 2$
tokens w.h.p.
\end{lemma}
\begin{proof}
Let
$U \subseteq V$
be the set of nodes that hold a token at time $t$.
For each node
$u \in U$,
we construct a rooted full binary tree $T_{u}$, augmented with functions
$\ell^{V}_{u}$
and
$\ell^{E}_{u}$
that assign labels
$\ell^{V}_{u}(x) \in V$
and
$\ell^{E}_{u}(x) \in E$,
respectively,
to each vertex $x$ of $T_{u}$.
The proof will proceed by coupling between the token passing process in $G$
and a process defined over the rooted trees $T_{u}$,
$u \in U$,
presented in the sequel.

Fix some node
$u \in U$
and let $\tau_{u}$ be the token held by $u$ at time $t$.
Notice that configuration $C^{t}$ determines the designated edge
$e = e_{i}(u) = \{ u, v \}$
of $u$,
$i \in \{ 1, 2, 3 \}$,
whose orientation is flipped (in round $t$ or
$t + 1$)
when $u$ disposes of $\tau_{u}$.
Add a new vertex $x$ as the root of $T_{u}$ and set
$\ell^{V}_{u}(x) \gets v$
and
$\ell^{E}_{u}(x) \gets e$.
The construction of $T_{u}$, $\ell^{V}_{u}$, and $\ell^{E}_{u}$ proceeds by
applying the following inductive rule to each newly added vertex $x$ with
labels
$\ell^{V}_{u}(x) = v$
and
$\ell^{E}_{u}(x) = e$:
\begin{itemize}

\item
Add new vertices $x_{a}$ and $x_{b}$ as children of $x$ in $T_{u}$ and set
$\ell^{V}_{u}(x_{j}) \gets w_{j}$
and
$\ell^{E}_{u}(x_{j}) \gets \{ v, w_{j} \}$
for
$j \in \{ a, b \}$
if
(I)
$\Degree(v) \geq 3$;
(II)
$e$ is two-sided designated;
and
(III)
$v$ is susceptible at time $t$ with
$I^{t}(v) = \{ \{ v, w_{a} \}, \{ v, w_{b} \}  \}$.

\item
Otherwise
($\Degree(v) < 3$,
$e$ is not two-sided designated,
or
$v$ is not susceptible at time $t$), make $x$ a leaf in $T_{u}$.

\end{itemize}

Let $S$ be the set of nodes that are susceptible at time $t$ and recall that
$\ell^{V}_{u}$ maps each internal vertex $x$ in $T_{u}$ to a node
$\ell^{V}_{u}(x) \in S$.
Since for each node
$v \in S$,
exactly one of the designated edges
$e_{i}(v)$,
$i \in \{ 1, 2, 3 \}$,
is oriented outward of $v$ at time $t$, it follows that there exists at most
one
$u \in U$
and at most one internal vertex
$x \in T_{u}$
such that
$\ell^{V}_{u}(x) = v$.
In other words, the restriction of the $\ell^{V}_{\cdot}$ functions to the
internal vertices in
$\bigcup_{u \in U} T_{u}$
is injective.
(This is in contrast to the restriction of the $\ell^{V}_{\cdot}$ functions to the
leaves in
$\bigcup_{u \in U} T_{u}$
that may map many leaves to the same node in $V$.)
We conclude that the total number of internal vertices in
$\bigcup_{u \in U} T_{u}$
is smaller than $n$ (recall that the nodes in $U$ are excluded from $S$),
hence
$\sum_{u \in U} |T_{u}| < 2 n$.

Consider a random experiment, referred to as the \emph{tree walk
process}, defined as follows:
Initially, a walker is placed on the root of $T_{u}$ for each
$u \in U$.
In every (discrete) step, if the walker of tree $T_{u}$ is placed on an
internal vertex
$x \in T_{u}$,
then the walker moves to a vertex
$x' \in T_{u}$
picked uniformly at random and independently among the two children of $x$ in
$T_{u}$.
The tree walk process in $T_{u}$ stops once the walker reaches a leaf.

Consider a node
$u \in U$
and an internal vertex
$x \in T_{u}$
and denote the size of the subtree of $T_{u}$ rooted at $x$ by $s(x)$.
Since $T_{u}$ is a (finite) full binary tree, it follows that
(1)
$s(y) < s(x)$
for each child $y$ of $x$ in $T_{u}$;
and
(2)
there exists a child $y$ of $x$ in $T_{u}$ such that
$s(y) < s(x) / 2$.
We conclude, by standard probabilistic arguments, that the walker in $T_{u}$
reaches a leaf within
$O (\log |T_{u}|) \leq O (\log n)$
steps w.h.p.

The key observation now is that we can couple between the tree walk process
over the trees $T_{u}$,
$u \in U$,
and the token passing process over $G$, running the former with the same
random source as the latter.
Specifically, we keep track of the token $\tau_{u}$ placed on node $u$ at time
$t$ and whenever a node
$v \in V$
disposes of $\tau_{u}$ by flipping the orientation of a designated edge
$e = \{ v, w \}$,
the walker of tree $T_{u}$ moves from a vertex
$x \in T_{u}$
with
$\ell^{V}_{u}(x) = v$
to a vertex
$x' \in T_{u}$
with
$\ell^{V}_{u}(x') = w$
and
$\ell^{E}_{u}(x') = e$.

The walker of tree $T_{u}$ reaching a leaf
$x \in T_{u}$
corresponds to one of the following three outcomes with respect to the token
$\tau_{u}$:
(I)
$\tau_{u}$ experiencing a token deletion outcome;
(II)
$\tau_{u}$ experiencing a token merging outcome;
or
(III)
$\tau_{u}$ experiencing a token shifting outcome into node
$\ell^{V}_{u}(x)$
over edge
$\ell^{E}_{u}(x)$
in some round
$t' > t$, however, this implies that another token $\tau_{\tilde{u}}$,
$\tilde{u} \neq u$,
experienced a token deletion outcome in some round
$t \leq \tilde{t} < t'$
when the orientation of a designated edge
$e_{i}(\ell^{V}_{u}(x)) \neq \ell^{E}_{u}(x)$,
$i \in \{ 1, 2, 3 \}$,
flipped from being oriented outward from
$\ell^{V}_{u}(x)$
to being oriented toward
$\ell^{V}_{u}(x)$
(thus making
$\ell^{V}_{u}(x)$
susceptible at time $t'$).
We conclude that once all walkers have reached the leaves of their
corresponding trees, the number of tokens in the graph have decreased by a
factor of at least $2$.
The assertion follows since each step of the tree walk process lasts (at most)
$2$ rounds in the token passing process.
\end{proof}

%%%%%%%%%%%%%%%%%%%%%%%%%%%%%%%%%%%%%%%%%%%%%%%%%%%%%%%%%%%%%%%%%%%%%%%%%%%%%%
\section{Maximal Node and Edge $k$-Coloring}
\label{sec:maximal-coloring}
%%%%%%%%%%%%%%%%%%%%%%%%%%%%%%%%%%%%%%%%%%%%%%%%%%%%%%%%%%%%%%%%%%%%%%%%%%%%%%
To simplify the following definition, denote
$X = V$
(resp.,
$X = E$).
For an integer
$k \geq 2$,
a \emph{maximal node $k$-coloring} (resp., \emph{maximal edge $k$-coloring})
\cite{Luby1986mis}
of the graph
$G = (V, E)$,
is a function
$f : X \rightarrow [k]$
that assigns a \emph{color}
$f(x) \in [k]$
to each node (resp., edge)
$x \in X$
and satisfies the following two conditions:
(I)
for each color
$1 \leq i < k$,
if nodes (resp., edges)
$x, y \in f^{-1}(i)$,
then $x$ and $y$ are not adjacent in $G$;
and
(II)
for each node (resp., edge)
$x \in f^{-1}(k)$
and for each color
$1 \leq i < k$,
there exists a node (resp., edge) $y$ adjacent to $x$ such that
$f(y) = i$.
Under the UP model, the problem of constructing a maximal node (resp., edge)
$k$-coloring, where
$k = O (1)$,
is translated to using $[k]$ as the output label set so that a node (resp.,
edge)
$x \in X$
is considered to be colored
$i \in [k]$
in a configuration $C$ if
$\omega(C(p)) = i$
for all ports incident on $x$.

\begin{theorem}
\label{thm:maximal-node-coloring}
For every constant integer
$k \geq 2$,
there exists a self-stabilizing UP algorithm that solves the maximal node
$k$-coloring problem and
stabilizes in $O(\log^{2} n)$ time w.h.p.
\end{theorem}

\begin{theorem}
\label{thm:maximal-edge-coloring}
For every constant integer
$k \geq 2$,
there exists a self-stabilizing UP algorithm that solves the maximal edge
$k$-coloring problem and
stabilizes in $O(\log^{5} n)$ time w.h.p.
\end{theorem}

To establish Theorem~\ref{thm:maximal-node-coloring} (resp.,
Theorem~\ref{thm:maximal-edge-coloring}), we employ the self-stabilizing MIS
(resp., MM) algorithm promised in Theorem~\ref{thm:mis} (resp.,
Theorem~\ref{thm:mm}), combined with the well known fact
(first observed in \cite{DolevIM1993dynamic-systems})
that the composition of self-stabilizing algorithms is also self-stabilizing.
Specifically, we run $k$ algorithms, denoted by
$\Alg_{1}, \dots, \Alg_{k}$,
where each $\Alg_{i}$ is invoked on the subgraph induced by the nodes (resp.,
edges) that are not colored by algorithms
$\Alg_{1}, \dots, \Alg_{i - 1}$.
For
$1 \leq i \leq k - 1$,
algorithm $\Alg_{i}$ runs the MIS algorithm of Theorem~\ref{thm:mis} (resp., MM
algorithm of Theorem~\ref{thm:mm}) and assigns the color $i$ to all nodes
(resp., edges) selected to be included in the constructed MIS (resp., MM).
Algorithm $\Alg_{k}$ simply assigns the color $k$ to all nodes (resp., edges)
in the subgraph it is invoked on.

To see why this works, let $S_{i}$ be the set of nodes (resp, edges) colored
by $\Alg_{i}$ and notice that for
$1 \leq i \leq k - 1$,
any two nodes (resp., edges) in $S_{i}$ cannot be adjacent as $S_{i}$ is an
independent set (resp., a matching).
The maximality of $S_{i}$,
$1 \leq i \leq k - 1$,
implies that if
$x \in S_{i + 1}$,
then $x$ has an adjacent node (resp., edge) in $S_{i'}$ for every
$1 \leq i' \leq i$.
The correctness of the composed maximal node (resp., edge) $k$-coloring
algorithm follows by plugging
$k = i + 1$.

The time it takes for the composed algorithm to stabilize is bounded from
above by $k$ times the runtime bound of each individual $\Alg_{i}$.
Theorem~\ref{thm:maximal-node-coloring} (resp.,
Theorem~\ref{thm:maximal-edge-coloring}) follows from Theorem~\ref{thm:mis}
(resp., Theorem~\ref{thm:mm}) as
$k = O (1)$.

%%%%%%%%%%%%%%%%%%%%%%%%%%%%%%%%%%%%%%%%%%%%%%%%%%%%%%%%%%%%%%%%%%%%%%%%%%%%%%
\section{The $2$-State MIS Process may be Slow}
\label{sec:two-state-process}
%%%%%%%%%%%%%%%%%%%%%%%%%%%%%%%%%%%%%%%%%%%%%%%%%%%%%%%%%%%%%%%%%%%%%%%%%%%%%%
The \emph{$2$-state MIS process} is arguably the simplest distributed node-centric (stone age) MIS algorithm.\footnote{%
Note that despite the simplicity of the node-centric $2$-state MIS process, we are unaware of a natural implementation of this algorithm under the (half-)edge-centric UP model.
Refer to Section \ref{sec:related-work} for further discussion of the obstacles in simulating stone age algorithms under the UP model.}
The automaton associated with each node has two states, referred to
as $\IN$ and $\OUT$, which are identified with the nodes' output for the MIS
problem in the natural manner.
An $\IN$ node is \emph{active} if it has some (at least one) $\IN$ neighbors;
an $\OUT$ node is \emph{active} if it has no $\IN$ neighbors.
A non-active node is referred to as \emph{passive}.
The algorithm is defined so that each active node updates its state by choosing a new state uniformly at random, while passive nodes stick to their current state.

It is easy to see that the $2$-state MIS process is guaranteed to stabilize to a
legal MIS w.p.\ $1$ on any (finite) graph
$G = (V, E)$,
starting from any initial configuration.
Giakkoupis and Ziccardi \cite{GiakkoupisZ2023two-state-process} prove that on
$G_{n, p}$,
the stabilization time is
$\log^{O (1)} (n)$
w.h.p.\ for ``most'' values of the parameter $p$, and conjecture that this
stabilization time upper bound holds, in fact, for any $n$-node graph $G$.
We refute this conjecture by proving the following theorem.

\begin{theorem}
\label{thm:two-state-process}
There exist constants
$c, c' > 0$
and an infinite family $\mathcal{G}$ of graphs
$G = (V, E)$
and corresponding initial configurations
$f : V \rightarrow \{ \IN, \OUT \}$
such that the stabilization time of the $2$-state MIS process on $G$ under $f$ is
$|V|^{c}$
w.p.\ at least
$c'$.
\end{theorem}

The construction of an $n$-node graph
$G = (V, E)$
in $\mathcal{G}$ relies on a gadget called
\emph{the
$(\alpha, \beta, \gamma)$-gadget},
where
$0 < \alpha < \beta < \gamma < 1$
are (rational) constants to be determined in the sequel;
for now, we merely require that $n$ is chosen so that
$n^{\alpha}$,
$n^{\beta}$,
and
$n^{\gamma}$
are all integers.
The
$(\alpha, \beta, \gamma)$-gadget
consists of $n^{\gamma}$ nodes with an edge connecting every two of them
(i.e., the gadget forms a clique).
Some $n^{\alpha}$ nodes of the gadget are regarded as \emph{$\alpha$-nodes};
out of the remaining
$n^{\gamma} - n^{\alpha}$
nodes of the gadget, some $n^{\beta}$ nodes are regarded as
\emph{$\beta$-nodes};
all remaining
$n^{\gamma} - n^{\beta} - n^{\alpha}$
nodes of the gadget are regarded as \emph{$\gamma$-nodes}.

The node set $V$ of the graph
$G = (V, E)$
is partitioned into
$k = n^{1 - \gamma}$
clusters, denoted by
$C_{1}, \dots, C_{k}$.
For each
$i \in [k]$,
the cluster $C_{i}$ forms an
$(\alpha, \beta, \gamma)$-gadget,
where we subsequently identify cluster $C_{i}$ with its corresponding gadget
and denote the sets of $\alpha$-, $\beta$-, and $\gamma$-nodes in $C_{i}$ by
$C_{i}^{\alpha}$, $C_{i}^{\beta}$, and $C_{i}^{\gamma}$, respectively.
This means, in particular, that every possible intra-cluster edge is present
in $G$.

The graph
$G = (V, E)$
is also augmented with the following inter-cluster edges:
For each
$1 \leq i < j \leq k$
and for each
$u \in C_{i}$
and
$v \in C_{j}$,
the edge
$\{ u, v \}$
is included in $E$ if and only if one of the following two conditions is
satisfied:
\\
(I)
$j = i + 1$
and
$u \in C_{i}^{\alpha}$;
or
\\
(II)
$j > i + 1$
and
$u \in C_{i}^{\alpha} \cup C_{i}^{\beta}$.
\\
Put differently, the inter-cluster edges form a complete bipartite graph
between the $\alpha$-nodes in $C_{i}$ and the nodes in
$\bigcup_{j = i + 1}^{k} C_{j}$;
and a complete bipartite graph between the $\beta$-nodes in $C_{i}$ and the
nodes in
$\bigcup_{j = i + 2}^{k} C_{j}$.
Notice that the $\gamma$-nodes in $C_{i}$ are not connected to (the nodes of)
$C_{j}$ for any
$j > i$.
The following observation is derived directly from the construction of $G$.

\begin{observation} \label{obs:mis-structure}
Consider an MIS
$U \subseteq V$
of $G$.
Then,
$|U \cap C_{1}| = 1$
and
$|U \cap C_{i}| \leq 1$
for every
$1 < i \leq k$.
Moreover, if
$|U \cap C_{j}^{\gamma}| = 1$
for each
$1 \leq j < i$,
then
$|U \cap C_{i}| = 1$.
\end{observation}

To complete the construction, we define the initial configuration
$f : V \rightarrow \{ \IN, \OUT \}$
associated with the graph
$G = (V, E)$.
To this end, we simply set
$f(v) = \IN$
if
$v \in C_{1}$;
and
$f(v) = \OUT$
otherwise.

Consider an execution of the $2$-state MIS process on $G$ under $f$ and the
$(\alpha, \beta, \gamma)$-gadget
associated with cluster $C_{i}$ for some
$1 \leq i \leq k$.
We say that the gadget is \emph{activated} in round $t$ of the execution if
\\
(I)
all nodes in the gadget are $\OUT$ and passive up to (including) round
$t - 1$;
and
\\
(II)
all nodes in the gadget become active, concurrently, in round $t$.
\\
Once the gadget is activated, the nodes start a \emph{tournament} so that each
node is $\IN$ for a (possibly empty) prefix of the tournament whose length is
a
$\GeomDist{1 / 2} - 1$
random variable;
following that, the node becomes $\OUT$ until the tournament ends.
A node \emph{wins} the tournament if it is the unique last $\IN$ node;
we show in the sequel that with sufficiently high probability, if the tournament admits a winner, then all nodes in the gadget become (passive and) stable.
If the tournament does not admit a (unique) winner, then a new tournament
starts, involving a subset $S$ of the nodes, where the content of $S$ depends
on events that occur at the rest of the graph.
We show in the sequel that with sufficiently high probability, this set $S$ includes all and only the gadget's $\gamma$-nodes.

Special attention is paid to the gadget's first tournament, referred to as the
\emph{principle tournament}, that starts once the gadget is activated (and
includes all nodes).
The principle tournament is said to be \emph{$\alpha$-completed} (if and) when
all the $\alpha$-nodes are $\OUT$, and \emph{$\beta$-completed} (if and) when
all the $\alpha$- and $\beta$-nodes are $\OUT$.
The principle tournament is said to be \emph{completed} when either all nodes
are $\OUT$ or when some node wins.

\begin{observation} \label{obs:principle-tournament}
For every choice of constants
$0 < \alpha < \beta < \gamma < 1$
and for every sufficiently small constant
$\epsilon > 0$,
there exists a constant
$c = c(\alpha, \beta, \gamma, \epsilon) > 0$
such that if an
$(\alpha, \beta, \gamma)$-gadget
is activated in round $t$, then its principle tournament is
\\
(1)
$\alpha$-completed by time
$t + \ell_{\alpha}$
for some
$(1 - \epsilon) \alpha \log n
\leq
\ell_{\alpha}
< (1 + \epsilon) \alpha \log n$
w.p.\ at least
$1 - n^{-c}$;
\\
(2)
$\beta$-completed by time
$t + \ell_{\beta}$
rounds for some
$(1 - \epsilon) \beta \log n
\leq
\ell_{\beta}
<
(1 + \epsilon) \beta \log n$
w.p.\ at least
$1 - n^{-c}$;
and
\\
(3)
completed by time
$t + \ell$
for some
$(1 - \epsilon) \gamma \log n
\leq
\ell
<
(1 + \epsilon) \gamma \log n$
w.p.\ at least
$1 - n^{-c}$.
\end{observation}

Assuming that cluster $C_{i}$ is activated (as defined earlier), let $T_{i}$
denote the corresponding principle tournament and let $t_{i}^{\alpha}$,
$t_{i}^{\beta}$, and $t_{i}$ be the times at which $T_{i}$ is $\alpha$-completed (assuming that it is  $\alpha$-completed), $\beta$-completed ((assuming that it is $\beta$-completed)), and completed, respectively.

By the definition of the initial configuration $f$, cluster $C_{1}$ is activated in round $0$, which means that tournament $T_{1}$ is well defined.
Condition hereafter on the events presented in Observation~\ref{obs:principle-tournament} with respect to $T_{1}$, which
means that $t_{1}^{\alpha}$, $t_{1}^{\beta}$, and $t_{1}$ are also well defined (the latter is actually well defined regardless).

The key observation now is that the construction of $G$ guarantees that
cluster $C_{2}$ is activated in round $t_{1}^{\alpha}$ (recall that $C_{2}$ is
dominated by each node in $C_{1}^{\alpha}$), hence tournament
$T_{2}$ is well defined.
Moreover, cluster $C_{j}$ cannot be activated before round
$t_{1}^{\beta}$
for any
$j > 2$
(recall that $C_{j}$ is dominated by each node in $C_{1}^{\beta}$).
Condition hereafter on the events presented in
Observation~\ref{obs:principle-tournament} with respect to $T_{2}$, which
means that $t_{2}^{\alpha}$, $t_{2}^{\beta}$, and $t_{2}$ are well defined
(the latter is actually well defined regardless).
By choosing
$\beta < 2 \alpha$
and
$\gamma > \alpha + \beta$,
we ensure that
\[
t_{1}^{\alpha}
<
t_{1}^{\beta}
<
t_{2}^{\alpha}
<
t_{2}^{\beta}
<
t_{1}
<
t_{2}
\, .
\]

As before, the construction of $G$ guarantees that cluster $C_{3}$ is
activated in round $t_{2}^{\alpha}$, hence tournament $T_{3}$ is well defined.
Moreover, cluster $C_{j}$ cannot be activated before round
$t_{2}^{\beta}$
for any
$j > 3$.
Condition hereafter on the events presented in
Observation~\ref{obs:principle-tournament} with respect to $T_{3}$, which
means that $t_{3}^{\alpha}$, $t_{3}^{\beta}$, and $t_{3}$ are well
defined (the latter is actually well defined regardless).
The choice of $\alpha$, $\beta$, and $\gamma$ ensures that
\[
t_{2}^{\alpha}
<
t_{2}^{\beta}
<
t_{3}^{\alpha}
<
t_{3}^{\beta}
<
t_{2}
<
t_{3}
\, .
\]

Taking
$h = \Omega (n^{c})$,
where
$c > 0$
is the constant promised in Observation~\ref{obs:principle-tournament} for a
sufficiently small $\epsilon$, we can
continue in this manner for
$i = 2, 3, \dots, h$
and conclude, by the union bound, that w.p.\
$\Omega (1)$,
cluster $C_{i}$ is activated in round
$t_{i - 1}^{\alpha}$
and
\[
t_{i - 1}^{\alpha}
<
t_{i - 1}^{\beta}
<
t_{i}^{\alpha}
<
t_{i}^{\beta}
<
t_{i - 1}
<
t_{i}
\]
for all
$1 < i \leq h$
(simultaneously).
In particular, there is a positive overlap between the round intervals
$[t_{i - 1}^{\alpha}, t_{i}^{\beta})$
and
$[t_{i}^{\alpha}, t_{i + 1}^{\beta})$
for every
$1 < i < h$.

The last statement plays a key role by ensuring that for every round
$t < t_{h}^{\beta}$,
there exists some
$1 \leq j \leq h$
such that some (at least one) nodes in $C_{j}^{\beta}$ are $\IN$ in round $t$.
Therefore, if cluster $C_{i}$ starts a non-principle tournament in round $t$
for some
$i \leq j - 2$,
then this tournament involves all and only the nodes in $C_{i}^{\gamma}$;
indeed, the $\beta$-nodes of $C_{j}$ (more accurately, those which are $\IN$)
prevent the $\alpha$- and $\beta$-nodes of $C_{i}$ from joining this
tournament.
(Combined with Observation~\ref{obs:mis-structure}, we conclude, in passing,
that if
$U \subseteq V$
is the output MIS, then
$|U \cap C_{i}^{\gamma}| = |U \cap C_{i}| = 1$
for each cluster $C_{i}$ that admits a winner before round
$t_{h}^{\beta}$.)
Theorem~\ref{thm:two-state-process} follows as the execution does not stabilize before time
$t_{h} \geq 3 h = \Omega (n^{c})$.

%%%%%%%%%%%%%%%%%%%%%%%%%%%%%%%%%%%%%%%%%%%%%%%%%%%%%%%%%%%%%%%%%%%%%%%%%%%%%%
\section{Additional Related Work}
\label{sec:related-work}
%%%%%%%%%%%%%%%%%%%%%%%%%%%%%%%%%%%%%%%%%%%%%%%%%%%%%%%%%%%%%%%%%%%%%%%%%%%%%%
%YE: discuss:
%- representing coloring via longest outgoing directed path
Assigning input/output labels from a finite set to the graph's half-edges (or
ports), rather than to the graph's nodes, is certainly not a new idea in
distributed computing.
For example, in the context of round elimination techniques for \emph{locally
checkable labeling (LCL)} problems \cite{NaorS1993lcl}, a formalism based on
such assignments was introduced by Brandt \cite{Brandt2019automatic} and
became a common practice since then.

The study of uniform distributed algorithms dates back to the classic work of
Angluin \cite{Angluin1980uniform} who proved that uniform Las Vegas algorithms
cannot elect a leader if termination detection is required (this impossibility
result was extended in \cite{ItaiR1990symmetry-breaking} to
uniform algorithms that are allowed to fail with a bounded probability).
In contrast, if the termination detection requirement is lifted, then leader
election is possible even under the (truly uniform) stone age model
\cite{VacusZ2025minimalist-le}.

The domain of self-stabilizing algorithms for local symmetry breaking problems
has recently seen a surge of activity, with many algorithms of various levels
of uniformity
\cite{Turau2019self-stabilizing,
EmekK2021thin,
GiakkoupisZ2023two-state-process,
BittonEIK2024mm,
GiakkoupisTZ2024beeping,
GiakkoupisTZ2025luby}.
However, none of the algorithms in those papers is both
(1)
truly uniform;
and
(2)
efficient on general graphs.

Truly uniform distributed computational models extend beyond the domain of
distributed graph algorithms.
One prominent example is the \emph{population protocols} model
\cite{AngluinADFP2006population-protocols}, or more generally \emph{chemical
reaction networks} \cite{SoloveichikCWB2008crn}, which abstract molecules in a
well-mixed solution.
Another prominent example is the \emph{geometric amoebot} model
\cite{DerakhshandehDGRS2014spaa-ba} for self-organizing particle systems
(a.k.a.\ programmable matter).

Returning to truly uniform distributed graph algorithms, the formulation of the stone age model referred to in the current paper is a simplified version, introduced in \cite{EmekK2021thin}, of the original stone age model formulation presented by Emek and Wattenhofer \cite{EmekW2013stone-age}.
In this regard, we note that the formulation of \cite{EmekW2013stone-age},
which is better suited for asynchronous message-passing schedules, involves
the notion of \emph{query letters};
this means that the signal received by an automaton at any given moment is
sensitive to only one type of message, determined by the automaton's current
state.
Emek and Wattenhofer \cite{EmekW2013stone-age} prove that the query letter restriction can be lifted, thus obtaining a formulation similar to that of \cite{EmekK2021thin}, however this proof assumes graceful initialization and it is not clear if it extends to the realm of self-stabilizing algorithms.

Since one of the results of Emek and Keren \cite{EmekK2021thin} is a self-stabilizing stone age MIS algorithm, it is interesting to compare it to the UP MIS algorithm developed in the current paper.
A crucial difference in this regard is that the MIS algorithm of \cite{EmekK2021thin} is applicable only to graphs of bounded diameter, i.e., the algorithm designer knows an upper bound
$D = O (1)$
on the graph diameter.
Using the knowledge of $D$, Emek and Keren developed some algorithmic machinery that essentially allows them to treat the input graph ``as if it was a complete graph''.
In particular, they use a probabilistic counting trick, developed originally for complete graphs, that generates globally agreed upon time intervals whose length is
$\Theta (\log n)$
w.h.p;
once the execution is partitioned into such intervals, breaking (local or global) symmetry becomes much easier. Unfortunately, such a synchronized partition of the execution cannot work in general graphs, so we had to use a completely different approach, as presented in Section~\ref{sec:mis}.

A natural question that arises from the results of the current paper is
whether there exists an efficient simulation of self-stabilizing stone age
algorithms (using the formulation of \cite{EmekK2021thin}) under the UP
model.\footnote{%
The converse direction is inherently impossible due to the added expressivity
of the UP model.}
Given the LLE mechanism that enables the selection of a designated port
$p^{*}(v) \in \Ports(v)$
for each node
$v \in V$
(see Section~\ref{sec:prelim}), one may hope that the simulation becomes
straightforward:
the designated port $p^{*}(v)$ is responsible for simulating the actions of
node $v$ in the simulated execution.

The caveat of this approach is that port $p^{*}(v)$ is not necessarily directly
exposed to the state of port $p^{*}(u)$ for the neighbors $u$ of $v$, thus it
takes up to $3$ rounds for node $v$ in the simulating execution to gather the
information regarding the states of its neighbors.
In an execution with a graceful initialization, this can be easily solved by
partitioning the simulating execution into $3$-round phases, each responsible
for one round of the simulated execution, so that the first $2$ rounds of a
phase are dedicated to gathering information about the states of the neighbors.
Unfortunately, this approach fails in the self-stabilization realm as the
phases of different nodes are not necessarily synchronized, leaving the
aforementioned question of simulating self-stabilizing stone age algorithms
under the UP model open.

On top of the node-centric stone age model, Emek and Wattenhofer also
introduce a stone age variant that includes both ``node devices'' and ``port
devices'' \cite{EmekW2013stone-age}[Section 4.4].
As the variant of \cite{EmekW2013stone-age}[Section 4.4] is also defined in
terms of query letters, it is not clear how it compares to the UP model
introduced in the current paper.
In any case, the formulation of the UP model is significantly simpler and in
our opinion, more natural.

%%%%%%%%%%%%%%%%%%%%%%%%%%%%%%%%%%%%%%%%%%%%%%%%%%%%%%%%%%%%%%%%%%%%%%%%%%%%%%
\clearpage
\bibliography{references}

%%%%%%%%%%%%%%%%%%%%%%%%%%%%%%%%%%%%%%%%%%%%%%%%%%%%%%%%%%%%%%%%%%%%%%%%%%%%%%
%%%%%%%%%%%%%%%%%%%%%%%%%%%%%%%%%%%%%%%%%%%%%%%%%%%%%%%%%%%%%%%%%%%%%%%%%%%%%%
\clearpage
\appendix

\begin{figure*}[!t]
{\centering
\Large{\textbf{APPENDIX}}
\par}
\end{figure*}

%%%%%%%%%%%%%%%%%%%%%%%%%%%%%%%%%%%%%%%%%%%%%%%%%%%%%%%%%%%%%%%%%%%%%%%%%%%%%%
\section{Proving Lemma~\ref{lem:prelim:probabilistic-progress}}
\label{apx:probabilistic-progress}
%%%%%%%%%%%%%%%%%%%%%%%%%%%%%%%%%%%%%%%%%%%%%%%%%%%%%%%%%%%%%%%%%%%%%%%%%%%%%%
As
$\Ex (X_{i} \mid X_{i - 1}) \leq r X_{i - 1}$
almost surely, we can apply the low of total expectation and conclude, by
induction on $i$, that
\[
\Ex (X_{i})
=
\Ex (\Ex (X_{i} \mid X_{i - 1}))
\leq
r \Ex (X_{i - 1})
\leq
r^{i} \Ex(X_{0})
=
r^{i} x_{0}
\, .
\]
Fixing
$\lambda = \lceil \log_{1 / r} x_{0} \rceil$,
we get
\[
\Ex (X_{\lambda + j})
\leq
r^{\lambda + j} x_{0}
\leq
r^{j}
\]
for every
$j \geq 0$.
As
$X_{\lambda + j}$
is integral and non-negative, Markov's inequality ensures that
\[
\Pr (X_{\lambda + j} > 0)
=
\Pr (X_{\lambda + j} \geq 1)
\leq r^{j}
\, .
\]
The assertion follows since
$X_{\lambda + j} > 0
\Longleftrightarrow
T > \lambda + j$.

\end{document}